\documentclass[11pt]{article}

\usepackage[super,comma,compress]{natbib}  
\usepackage{lineno}
\usepackage{subfig}
\makeatletter 
\renewcommand\@biblabel[1]{#1.} 
\makeatother

\usepackage{times}

\usepackage{amsmath}
\usepackage{amsfonts}

\usepackage{booktabs}
\usepackage{dcolumn}

\usepackage{enumitem}

\usepackage{rotating}
\usepackage{graphicx}

\newtheorem{lemma}{Lemma}
\newtheorem{proof}{Proof}
\newtheorem{proposition}{Proposition}
\newtheorem{theorem}{Theorem}
\newtheorem{definition}{Definition}
\newtheorem{corollary}{Corollary}

\usepackage{setspace}

\newenvironment{sciabstract}{ 
\begin{quote} \bf }
{\end{quote}}

\newcounter{lastnote}

\usepackage{color}
\definecolor{Blue}{rgb}{0,0,1}

\definecolor{Orange}{rgb}{1,0.5,0}

\definecolor{Green}{rgb}{0,1,0}

\usepackage{float}

\vspace{-2cm}
\title{\vspace{-2cm}
Contextual centrality: going beyond network structure}

\vspace{.2cm}

\author
{
 Yan Leng $^{1\dagger}$, Yehonatan Yella$^{2}$, Rodrigo Ruiz$^{1}$, Alex Pentland$^{1}$  \vspace{.25cm} \\
\footnotesize{$^{1}$Massachusetts Institute of Technology, Cambridge, MA, USA} \vspace{-.15cm} \\
\footnotesize{$^{2}$Albert Einstein College of Medicine, New York, USA} \vspace{-.15cm} \\
\footnotesize{$^\dagger$ To whom correspondence should be addressed; E-mail:  yleng@mit.edu} 
}
\vspace{1cm}

\date{}

\usepackage[british]{babel}
\usepackage{enumitem}

\newlist{SubItemList}{itemize}{1}
\setlist[SubItemList]{label={$-$}}

\let\OldItem\item
\newcommand{\SubItemStart}[1]{%
    \let\item\SubItemEnd
    \begin{SubItemList}[resume]%
        \OldItem #1%
}
\newcommand{\SubItemMiddle}[1]{%
    \OldItem #1%
}
\newcommand{\SubItemEnd}[1]{%
    \end{SubItemList}%
    \let\item\OldItem
    \item #1%
}
\newcommand*{\SubItem}[1]{%
    \let\SubItem\SubItemMiddle%
    \SubItemStart{#1}%
}%

\begin{document} 

\baselineskip24pt


\maketitle

\vspace{0mm}
\textbf{
Centrality is a fundamental network property which ranks nodes by their structural importance. However, this alone may not predict successful diffusion in many applications, such as viral marketing and political campaigns. We propose contextual centrality, which integrates structural positions, the diffusion process, and, most importantly, relevant node characteristics. It nicely generalizes and relates to standard centrality measures. We test the effectiveness of contextual centrality in predicting the eventual outcomes in the adoption of microfinance and weather insurance. Our empirical analysis shows that the contextual centrality of first-informed individuals has higher predictive power than that of other standard centrality measures. Further simulations show that when the diffusion occurs locally, contextual centrality can identify nodes whose local neighborhoods contribute positively. When the diffusion occurs globally, contextual centrality signals whether diffusion may generate negative consequences.  Contextual centrality captures more complicated dynamics on networks and has significant implications for network-based interventions.
}

Individuals, institutions, and industries are increasingly connected in networks where the behavior of one individual entity may generate a global effect \cite{jackson2010social, liu2017industrial,acemoglu2015systemic}. Centrality is a fundamental network property which captures an entity's ability to impact macro processes, such as information diffusion on social networks \cite{jackson2010social}, cascading failures in financial institutions \cite{acemoglu2015systemic}, and the spreading of market inefficiencies across industries \cite{liu2017industrial}. Many interesting studies have found that the structural positions of individual nodes in a network explain a wide range of behaviors and consequences. Degree centrality predicts who is the first to be infected in a contagion \cite{christakis2010social}. Eigenvector centrality corresponds to the incentives to maximize social welfare \cite{galeotti2017targeting}. Katz centrality is proportional one's power in strategic interactions in network games \cite{ballester2006s}. Diffusion centrality depicts an individual's capability of spreading in information diffusion \cite{banerjee2014gossip}. These centrality measures operate similarly, aiming to reach a large crowd via diffusion, and are solely dependent on the network structure. 

However, several pieces of empirical evidence show that reaching a large crowd may generate adverse effects. For example, sales on Groupon, public announcements of popular items on Goodreads, and online video platforms are successful in reaching a large population of customers. However, studies show that these strategies lower online reviews by reaching the population of people who hold negative opinions about the product \cite{byers2012groupon, kovacs2014paradox, aizen2004traffic}. Let us further consider two motivating examples to demonstrate the importance of accounting for node characteristics.

\emph{Example 1. Viral marketing. } During a viral marketing campaign, the marketing department aims to attract more individuals to adopt the focal product. If we have ex-ante information about the customers' likelihood of adoption, we can utilize it to better target individuals who have higher chances of adoption and avoid wasting resources on those otherwise. 

\emph{Example 2. Political campaign.}  Typical Get-Out-The-Vote (GOTV) campaigns include direct mail, phone calls,  and social-network advertisement \cite{green2015get,bond201261}. However, rather than merely aiming to transform nonvoters to voters,  a GOTV strategy should target voters  who are more likely to vote for the campaigner's candidate.

In this paper, we introduce contextual centrality, which builds upon diffusion centrality proposed in Banerjee et al. and  captures relevant node characteristics in the objective of the diffusion \cite{banerjee2013diffusion,banerjee2014gossip}. Moreover, contextual centrality aggregates these characteristics over one's neighborhood, which is defined by the diffusion process. It generalizes and nests degree, eigenvector, Katz, and diffusion centrality. When the spreadability (the product between the diffusion rate $p$ and the largest eigenvalue $\lambda_1$ of the adjacency matrix) and the diffusion period $T$ are large,  contextual centrality linearly scales with eigenvector, Katz, and diffusion centrality. The sign of the scale factor is determined by the joint distribution of nodes' contributions to the objective of the diffusion and their corresponding structural positions. 

We perform an empirical analysis of the diffusion of microfinance and weather insurance showing that the contextual centrality of the first-informed individuals better predicts the adoption decisions than that of the other centrality measures mentioned above. Moreover, simulations on the synthetic data show how network properties and node characteristics collectively influence the performance of different centrality measures. Further, we illustrate the effectiveness of contextual centrality over a wide range of diffusion rates with simulations on the real-world networks and relevant node characteristics in viral marketing and political campaigns. 


\section*{Contextual centrality}

Given a set of $N$ individuals, the adjacency matrix of the network is $\mathbf{A}$, an $N$-by-$N$ symmetric matrix. The entry ${A}_{ij}$ equals 1 if there exists a link between node $i$ and node $j$, and 0 otherwise.  Let $\mathbf{D} = \text{diag}(\mathbf{d})$, where $d_i = \sum_{j=1}^N {A}_{ij}$ denotes the degree of node $i$. With Singular Value Decomposition, we have $\mathbf{A} =\mathbf{U} \boldsymbol{\Lambda} \mathbf{U}^{T}$, where $\boldsymbol{\Lambda}  =  \text{diag} \{ \boldsymbol{\Lambda} \}= \{\lambda_1, \lambda_2, ...,\lambda_n\}$ in a descending order and the corresponding eigenvectors are $\{\mathbf{U}_1, \mathbf{U}_2, ..., \mathbf{U}_n\}$ with  $\mathbf{U}_1$ being the leading eigenvector. We let $\circ$ denote the Hadamard product (i.e., element-wise multiplication). We use bold lowercase variables for vectors and bold upper case variables for matrices. 

The diffusion process in this paper follows the independent cascade model \cite{kempe2003maximizing}. It starts from an initial active seed. When node $u$ becomes active, it has a single chance to activate each currently inactive neighbor $v$ with probability ${P}_{uv}$, where $\mathbf{P} \in \mathbb{R}^{N \times N}$. We follow the terminology by Koschutzki to categorize degree, eigenvector, and Katz centrality as reachability-based centrality measures \cite{koschutzki2005centrality}. 
Reachability-based centrality measures aim to score a certain node $v$ by the expected number of individuals activated if $v$ is activated initially, $s (v, \mathbf{A}, \mathbf{P})$, and hence tend to rank nodes that can reach more nodes in the network higher. In particular, 
\begin{equation}
 s (v, \mathbf{A}, \mathbf{P} )  =  \sum_{i}^N r_i (v,  \mathbf{A}, \mathbf{P} ),
\label{obj:cascade_size}
\end{equation}
where ${r}_i (v, \mathbf{A}, \mathbf{P})$ denotes the probability that $i$ is activated if $v$ is initially activated\cite{kempe2003maximizing, wen2017online, lynn2016maximizing}. In practice, $ s (v, \mathbf{A}, \mathbf{P} )$ is hard to estimate. Different reachability-based centrality measures estimate it in different ways. 

Diffusion centrality extends and generalizes these standard centrality measures \cite{banerjee2013diffusion}. It operates on the assumption that the activation probability of an individual $i$ is correlated with the number of times $i$ `hears' the information originating from the individual to be scored. Diffusion centrality measures how extensively the information spreads as a function of the initial node \cite{banerjee2013diffusion}. In other words, diffusion centrality scores node $v$ by the expected number of times  some piece of information originating from $v$ is heard by others within a finite number of time periods $T$ , $s' (v, \mathbf{A},\mathbf{P}, T)$,  
\begin{equation}
s' (v, \mathbf{A},\mathbf{P}, T)  =  \sum_{i}^N r'_i (v, \mathbf{A}, \mathbf{P}, T ),
\label{obj:cascade_size2}
\end{equation}
where $r'_i (v,\mathbf{A}, \mathbf{P}, T)$ is the expected number of times individual $i$ receives the information if $v$ is seeded. Eq.~(\ref{obj:cascade_size2}) has at least two advantages over Eq.~(\ref{obj:cascade_size}). First, $r'_i (v,\mathbf{A}, \mathbf{P}, T)$ is computationally more efficient than tedious simulations to get $r_i (v,\mathbf{A}, \mathbf{P})$. Second, it nests degree, eigenvector, and Katz centrality \cite{banerjee2014gossip}\footnote{
It is worth noting that Eq.~(\ref{obj:cascade_size}) and (\ref{obj:cascade_size2}) differ in a couple of ways.  First, since $r'_i (v,  \mathbf{A}, \mathbf{P}, T)$ is the expected number of times $i$ hears a piece of information, it may exceed 1. Meanwhile, since $r_i (v, \mathbf{A}, \mathbf{P})$ is the probability that $i$ receives the information, it is bounded by 1. Second, in independent cascade model, each activated individual has a single chance to activate the non-activated neighbors. However, with the random walks of information transmission in contextual centrality, each activated individual has multiple chances with decaying probability to activate their neighbors.}. 

Both Eq.~(\ref{obj:cascade_size}) and (\ref{obj:cascade_size2}) assume that individuals are homogeneous and contribute equally to the objectives if they have been activated. However, in many real-world scenarios, such as the two examples mentioned above, the payoff for the campaigner does not grow with the size of the cascade. Instead, different nodes contribute differently. Formally, let $y_i$ be the contribution of individual $i$ to the cascade payoff upon being activated. Note that  $y_i$ is context-dependent and is measured differently in different scenarios. For example, in a market campaign, $y_i$ can be $i$'s likelihood of adoption. In a political campaign, $y_i$ can be the likelihood that $i$ votes for the political campaigner's political party. With the independent cascade model, an individual $v$ should be scored according to the cascade payoff if $v$ is first-activated, $ s_c (v, \mathbf{A}, \mathbf{p})$. With this, we present the following equation as a generalization and extension to Eq.~(\ref{obj:cascade_size}) with heterogeneous $\mathbf{y}$, 
\begin{equation}
\text{cascade payoff:\;\;}s_c (v, \mathbf{A}, \mathbf{p})  = \sum_{i}^N r_i (v, \mathbf{A}, \mathbf{P} ) y_i. 
\label{obj:cascade_payoff}
\end{equation}

Similar to  diffusion centrality, we score nodes with the following approximated cascade payoff, $s'_c (v, \mathbf{A}, \mathbf{p}, T)$, with heterogeneous $\mathbf{y}$, 
\begin{equation}
 \text{approximated cascade payoff:\;\;}s'_c (v, \mathbf{A}, \mathbf{P}, T)  = \sum_{i}^N r'_i (v, \mathbf{A}, \mathbf{P}, T) y_i. 
\label{obj:cascade_payoff_2}
\end{equation}

This formulation generalizes diffusion centrality and inherits its nice properties in nesting existing reachability-based centrality measures. Moreover, it is more easier to compute than Eq.~(\ref{obj:cascade_payoff})\footnote{The computational complexity of the algorithm to score according to Eq.~(\ref{obj:cascade_payoff}) is $O(NM^T)$, where $M$ is the average degree and $T$ is the paths lengths that has been inspected. Note that the computational complexity of the formulation~(\ref{ic_def}) is $O(N M T)$. We repeat the operation of multiplying a vector of length $N$ with a sparse matrix which has an average of $M$ entries per row for $T$ times. This significantly reduces the run time.}. With this scoring function, we now formally propose contextual centrality. 
\begin{definition}
Contextual centrality (CC) approximates the cascade payoff  within a given number of time periods $T$ as a function of the initial node accounting for individuals' contribution to the cascade payoff.  
\begin{equation}
\text{CC}(\mathbf{A}, \mathbf{P}, T, \mathbf{y})  :=\sum_{t=0}^{T}(\mathbf{P} \circ \mathbf{A})^{t}\mathbf{y}, 
\label{ic_def}
\end{equation}
\end{definition}

Heterogeneous diffusion rates across individuals are difficult to collect and estimate in real-world applications. Therefore, in the following analysis, we assume a homogeneous diffusion rate $p$ across all edges. Hence, $\mathbf{P} \circ \mathbf{A}$ in Eq.~(\ref{ic_def}) is reduced to ${p} \mathbf{A}$. Similar to diffusion centrality, contextual centrality is a random-walk-based centrality, where $(p \mathbf{A})^t$ measures the expected number of walks of length $t$ between each individual pair and $T$ is the maximum walk-length considered. Since $T$ is the longest communication period, larger $T$ indicates a longer period for diffusion (e.g., a movie that stays in the market for a long period) while smaller $T$ indicates a shorter diffusion period (e.g., a coupon that expires soon). 
One the one hand, when $p\lambda_1$ is larger than $1$,  CC approaches infinity as $T$ grows. On the other hand, when $p\lambda_1 < 1$, CC is finite for $T=\infty$, which can be understood as a lack of virality, expressed in a fizzling out of the diffusion process with time. In fact, the specific value of $p \lambda_1$ can be used to bound the maximum possible CC given the norm of the score vector $\textbf{y}$. As presented in proposition \ref{ccineq} in the Supporting Information, the upper bound for CC grows with $p \lambda_1$ and the norm of the score vector. 

Let us further illustrate the relationship between CC and diffusion centrality, DC for short\footnote{In Banerjee et al.\cite{banerjee2013diffusion}, $\text{DC} = \sum_{t=1}^T (p \mathbf{A} )^t$. To derive the following relationship between CC and DC, we add the score of reaching the first seeded individual into computing diffusion centrality. Hence, $\text{DC} = \sum_{t=0}^T (p \mathbf{A} )^t$. Adding the first seeded individual into the scoring function produces the same ranking as the one used in Banerjee et al. }. We can represent $\mathbf{y}$ as, $\mathbf{y} = \sigma(\mathbf{y}) \cdot \mathbf{z} + \overline{\mathbf{y}} \cdot \mathbf{1}$, where $\sigma(\mathbf{y})$ and $\mathbf{z}$ are the standard deviation and the z-score normalization of $\mathbf{y}$. Using the linearity of CC with respect to $\mathbf{y}$, we can write
\begin{equation}
\begin{split}
\text{CC}(\mathbf{A}, p, T, \mathbf{y}) & = \sigma(\mathbf{y}) \cdot \text{CC}(\mathbf{A}, p, T, \mathbf{z}) + \overline{\mathbf{y}} \cdot \underbrace{\text{CC}(\mathbf{A}, p, T, \mathbf{1}}_\text{DC}) 
\label{eq:std_mean_tradeoff}
\end{split}
\end{equation}

Eq.~(\ref{eq:std_mean_tradeoff}) shows the trade-off between the standard deviation $\sigma(\mathbf{y})$ and the mean $\overline{\mathbf{y}}$ of the contribution vector in CC.  When $\overline{\mathbf{y}}$ dominates over $\sigma(\mathbf{y})$, network topology is more important in CC and it produces similar rankings to DC. If the graph is Erdos-Renyi and T is small enough, then, on expectation, the term $\overline{\mathbf{y}} \cdot \text{DC}$ dominates as the size of the network approaches infinity, as presented in Theorem~\ref{expccdc} in the Supporting Information.  However, when $\sigma(\mathbf{y})$ dominates over $\overline{\mathbf{y}}$, CC and DC generate very different rankings. 

The relevant node characteristics $\mathbf{y}$ provides the ex-ante estimation about one's contribution. Whether to incorporate $\mathbf{y}$ is the main difference between our centrality and existing centrality measures. In the real-world data, the observation or estimation on $\mathbf{y}$ can be noisy, biased or stochastic. Therefore, we discuss the robustness of contextual centrality in responses to perturbations in $\mathbf{y}$ in the Supporting Information.

We define the following terms, which we use throughout the paper:
\begin{itemize}[topsep=0pt,itemsep=-1ex,partopsep=1ex,parsep=1ex]
    \item Spreadability $(p \lambda_1)$ captures the capability of the campaign to diffuse on the network depending on the diffusion probability ($p$) via a certain communication channel, and the largest eigenvalue ($\lambda_1$) of the network.
     \item Standardized average contribution $(\frac{\overline{\mathbf{y}}}{\sigma (\mathbf{y})})$  is computed as the average of the contributions normalized by the standard deviation of the contribution. The sign of $\frac{\overline{\mathbf{y}}}{\sigma (\mathbf{y})}$ indicates whether the average contribution is positive or not. Moreover, the larger the magnitude of $\frac{\overline{\mathbf{y}}}{\sigma (\mathbf{y})}$, the more homogeneous the contributions are. 
    \item Primary contribution $(\mathbf{U}_1^T \mathbf{y})$ measures the joint distribution of the structural importance and nodal contributions. It captures whether people who dominate important positions have positive contributions or not. 
\end{itemize}

\paragraph{Properties of contextual centrality when $p \lambda_1 > 1$ and $T$ is large.} Let us first provide the approximation of contextual centrality in this condition, which reveals one of the  prominent advantages of contextual centrality. 
By the Perron-Frobenius Theorem, we have $|\lambda_j|\le \lambda_1$ for every $j$. Moreover, if we assume that the graph is non-periodic, then in fact $|\lambda_j|<\lambda_1$ for all $j \ne 1$. Note that the typical random graph is not periodic, so this assumption is reasonable. Thus, when $p \lambda_1 > 1$, the term $(p\lambda_1)^t$ grows exponentially faster than $(p\lambda_j)^t$ for $j \ne 1$ so that the $j=1$ term dominates for sufficiently large values of $T$, and we obtain the approximation for contextual centrality (CC$_\text{approx}$):
\begin{equation}
\begin{split}
\text{CC}  =\sum_{j=1}^n \sum_{t=0}^T (p\lambda_j)^t \mathbf{U}_j\mathbf{U}_j^T\mathbf{y} \approx \text{CC}_\text{approx} = \Big( \sum_{t=0}^T (p\lambda_1)^t \mathbf{U}_1^T \mathbf{y} \Big) \mathbf{U}_1.
\label{eq:prox_cc}
\end{split}
\end{equation}

This approximation reveals some desirable properties of contextual centrality. Crucially, CC$_\text{approx}$ is simply a scalar multiple of the leading eigenvector when $p \lambda_1 > 1$ and $T$ is large. Therefore, the sign of $\mathbf{U}_1^T \mathbf{y}$ determines the direction of the relationship between CC$_\text{approx}$ and eigenvector centrality. By Perron-Frobenius Theorem, all elements in this leading eigenvector are nonnegative. Thus, the approximated cascade payoff, Eq.~(\ref{obj:cascade_payoff_2}), for seeding any individual is nonpositive if $\mathbf{U}_1^T \mathbf{y} < 0$, $p \lambda > 1$ and $T$ is large. This shows that in this condition the approximated cascade payoff is nonpositive for seeding any individual, so the campaigner should select a diffusion channel with lower diffusion rate to take advantage of the local neighborhood with positive contributions. Eq.~(\ref{eq:prox_cc}) naturally suggests the following relationships between $\text{CC}_\text{approx}$ and eigenvector centrality. 
    \begin{itemize}[topsep=0pt,itemsep=-1ex,partopsep=1ex,parsep=1ex]
        \item If $\mathbf{U}_1^T \mathbf{y} > 0$, CC$_\text{approx}$ and eigenvector centrality produce the same rankings. 
        \item If $\mathbf{U}_1^T \mathbf{y} < 0$, CC$_\text{approx}$ and eigenvector centrality produce the opposite rankings. 
    \end{itemize}        
The approximation does not hold when $\mathbf{U}_1^T \mathbf{y} = 0$, which is also unlikely to happen in practice. Hence, we disregard this case. Similarly, we relate contextual centrality to diffusion centrality ($\text{C}_\text{Diffusion}$) and Katz centrality ($\text{C}_\text{Katz}$), 
\begin{equation}
\begin{split}
& \text{C}_{\text{Diffusion}} \propto {\sum_{t=1}^{\infty}( p \lambda_1)^{t}}\mathbf{U}_{1}{(\mathbf{U}_{1}^{T}\mathbf{1})} = \frac{\sum_{t=1}^{\infty}( p \lambda_1)^{t}(\mathbf{U}_{1}^{T}\mathbf{1}) }{\sum_{t=0}^T (p\lambda_1)^t \mathbf{U}_1^T \mathbf{y}} \text{CC}_\text{approx},  \\
& \text{C}_{\text{Katz}} \propto {\sum_{t=0}^{\infty}( \alpha \lambda_1)^{t}}\mathbf{U}_{1}{(\mathbf{U}_{1}^{T}\mathbf{1})}= \frac{\sum_{t=0}^{\infty}( \alpha \lambda_1)^{t}(\mathbf{U}_{1}^{T}\mathbf{1}) }{\sum_{t=0}^T (p\lambda_1)^t \mathbf{U}_1^T \mathbf{y}} \text{CC}_\text{approx}, 
\end{split}
\label{eq:prox_cc_2}
\end{equation}
where $\alpha$ is the decay factor in Katz centrality. Similar to Eq.~(\ref{eq:prox_cc}), all terms on the right-hand-side of Eq.~(\ref{eq:prox_cc_2}) are positive except for $\mathbf{U}_1^T \mathbf{y}$, which similarly determines the direction of the relationship.

\subsection*{Results}


\paragraph{Predictive power of contextual centrality} 


We study two real-world empirical settings, adopting microfinance in 43 Indian villages\footnote{The data is made public by Banerjee et al. \cite{banerjee2013diffusion}.} and adopting weather insurance in 47 Chinese villages\footnote{The data is made public by Cai et al. \cite{cai2015social}.}. In each setting, there is a set of first-informed households in each village who went on to spread the information. We evaluate the adoption outcome of all other households in the village which are not first-informed. We use the adoption likelihood for the contribution vector $\mathbf{y}$ in computing contextual centrality, which is predicted using a model based on the adoption decisions of the first-informed households. Similar to Banerjee et al. \cite{banerjee2013diffusion}, we evaluate the $R^2$ of a linear regression model, where the independent variable is the average centrality of first-informed households and the dependent variable is the fraction of non-first-informed households in a village which adopted - controlled by the village size. In Fig.~\ref{fig:r2}, we show how the $R^2$ for various measures of centrality varies with $p \lambda_1$, in which the choice of $p$ influences the two centrality measures that account for the diffusion process -  diffusion centrality and contextual centrality. We see that the contextual centrality outperforms all other standard centrality measures, which indicates that marketing campaigners or social planners will benefit from using contextual centrality as the seeding strategy to maximize the participation. This result also highlights that utilizing ex-ante information about customers' likelihood of adoption helps to design better targeting strategies. Similar results without control variables and with more control variables are presented in the Supporting Information as a robustness check.

\begin{figure*}
\centering
\subfloat[Adoption of microfinance in Indian villages]{\includegraphics[width=.5\linewidth]{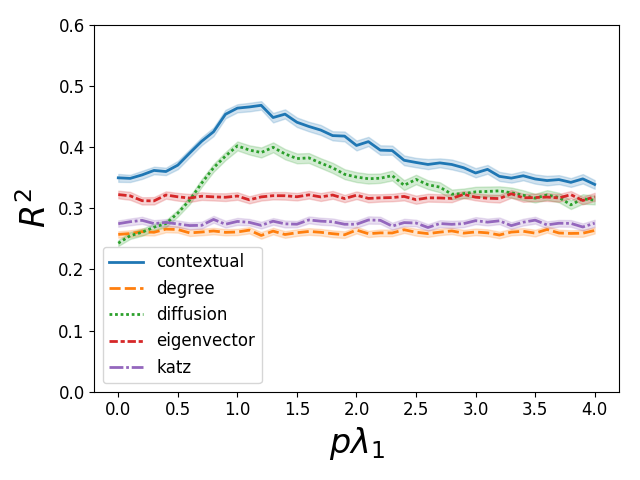}\label{fig:r2_mf}}
\subfloat[Adoption of weather insurance in Chinese villages]{\includegraphics[width=.5\linewidth]{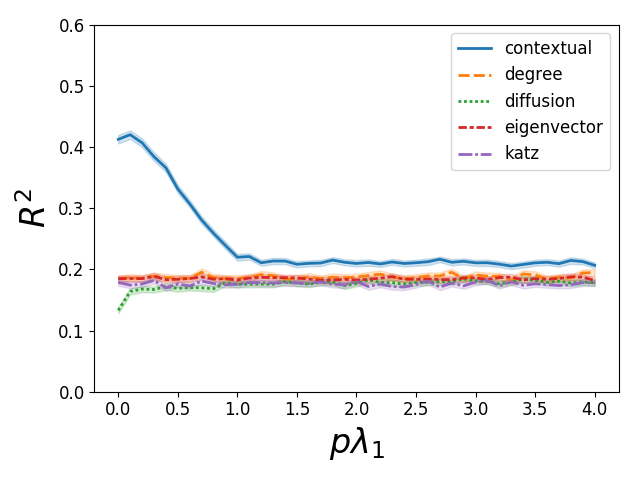}\label{fig:r2_insr}}
\caption{Predictive power of contextual centrality. We show how the average centrality of first-informed individuals predicts the eventual adoption rate of non-first-informed individuals in (a) microfinance and (b) weather insurance. The y-axis shows the 95\% confidence interval of $R^2$ computed from 1000 bootstrap samples from ordinary least squares regressions controlling for village size. The x-axis shows varying values for $p \lambda_1$, which influences only diffusion centrality and contextual centrality.}
\label{fig:r2}
\end{figure*}

\paragraph{Performance of contextual centrality relative to other centralities on random networks}

To better understand CC's performance with respect to different parameters ($p \lambda_1$, $\frac{\overline{\mathbf{y}}}{\sigma (\mathbf{y})}$), we next perform simulations on randomly generated synthetic networks and contribution vectors ($\mathbf{y}$). To compare the performance of contextual centrality against the other centralities, we use `relative change'  (calculated as $\frac{a - b}{max(|a|, |b|)}$, where $a$ is a given centrality's average payoff and $b$ is the maximum average payoff of the other centralities)\footnote{We chose `relative change' for comparison since it gives a sense of when the payoffs are different from the optimal centrality while keeping the magnitudes of the payoffs in perspective. This measure has some desirable properties. First, its value is necessarily between -2 and 2, so our scale for comparison is consistent across scenarios. Second, its magnitude does not exceed 1 unless $a$ and $b$ differ in sign, so we can tell if a centrality gets a positive average payoff while the rest do not.}. Fig.~\ref{fig:perform_cc_simulated} displays the relative change between CC's average payoff and the maximum average payoff of the other centralities aggregated over 100 runs of simulations for varying values of $\frac{\overline{\mathbf{y}}}{\sigma(\mathbf{y})}$ and $p \lambda_1$ on three different types of simulated graphs. We can see that CC performs well when $\overline{\mathbf{y}}   < 0$, $p \lambda_1 < 1$, and $\frac{\overline{\mathbf{y}}}{\sigma(\mathbf{y})}$ is small in magnitude. We will now discuss each of these cases in more detail.

When $\overline{\mathbf{y}} < 0$, maximizing the reach of the cascade is not ideal because that will result in a cascade payoff which more closely reflects $\overline{\mathbf{y}}$. CC differs from the other centralities in that it does not try to maximize the reach of the cascade. Note the dark blue diagonal band present in all plots in Fig.~\ref{fig:perform_cc_simulated}. Since the magnitude of the relative change exceeds 1 only when the values being compared have opposite signs, this region shows that there are many settings where the standard average contribution is negative, yet CC achieves a positive payoff while the other centralities do not. 

When $p \lambda_1$ is small, it is essential to seed an individual whose local neighborhood has higher individual contributions since there is not much risk of diffusing to individuals with lower individual contributions\footnote{As an extreme case, consider $p \lambda_1 = 0$. In this case, the diffusion rate is 0, so seeding an individual with a high individual payoff makes much more sense than seeding an individual with high topological importance.}. This highlights CC's advantage in discriminating the local neighborhoods with positive payoffs from those with negative payoffs while the other centralities cannot. 

When $\frac{\overline{\mathbf{y}}}{\sigma(\mathbf{y})}$ is small in magnitude, CC takes advantage of the greater relative variations between contributions. As $\frac{\overline{\mathbf{y}}}{\sigma(\mathbf{y})} \rightarrow +\infty$, Eq.~(\ref{eq:std_mean_tradeoff}) tells us that CC will seed similar to DC, which explains why CC loses some of its advantage. However, as $\frac{\overline{\mathbf{y}}}{\sigma(\mathbf{y})} \rightarrow -\infty$, Eq.~(\ref{eq:std_mean_tradeoff}) tells us that CC will seed opposite to DC, which explains why CC maintains an advantage.

We briefly comment on the regions where CC does not seem to offer an advantage. CC appears to do slightly worse when $\overline{\mathbf{y}} > 0$ and $p \lambda_1$ is a bit greater than 1. We expect that when $p \lambda_1 > 1$, CC would offer less of an advantage since the cascade reaches most individuals in the network. However,  Eq.~(\ref{eq:prox_cc}) (and to some extent Eq.~(\ref{eq:std_mean_tradeoff})) suggests that CC should seed similar to the other centralities. Note that in Fig.~\ref{fig:relative_change_synthetic}, which averages over the samples of all three graph types, CC performs better. Thus, we conclude that this is due to the high variance of the cascade payoffs in this region. CC also seems to perform poorly when $\overline{\mathbf{y}} = 0$. This is because as $p \lambda_1$ increases, the payoff more closely reflects $\overline{\mathbf{y}}$, which means that CC's average payoff and the maximum average payoff of the other centralities are very close to 0 and often have different signs, so the relative change would appear to be large. In these two regions where CC does not seem to offer an advantage, no single centrality dominates the rest, which shows that there are large relative variations in these regions. Similar figures to Fig.~\ref{fig:perform_cc_simulated} for all other centrality measures used in this paper can be found in the Supporting Information. 

\begin{figure*} 
\centering
\subfloat[Barabasi-Albert model]{\includegraphics[width=.4\linewidth]{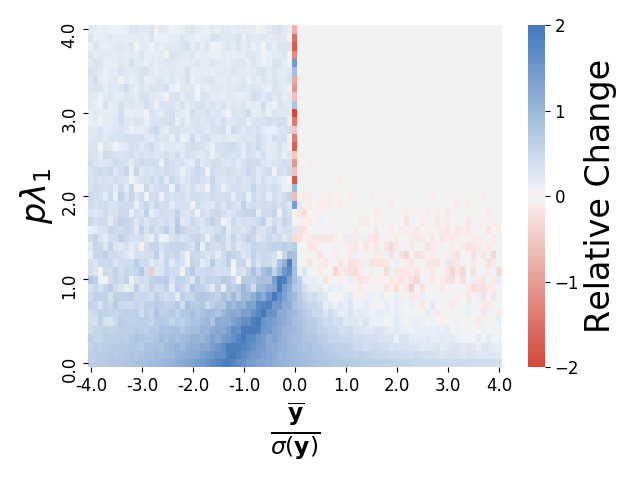} \label{fig:relative_change_barabasi_albert}}
\subfloat[Erdos-Renyi model]{\includegraphics[width=.4\linewidth]{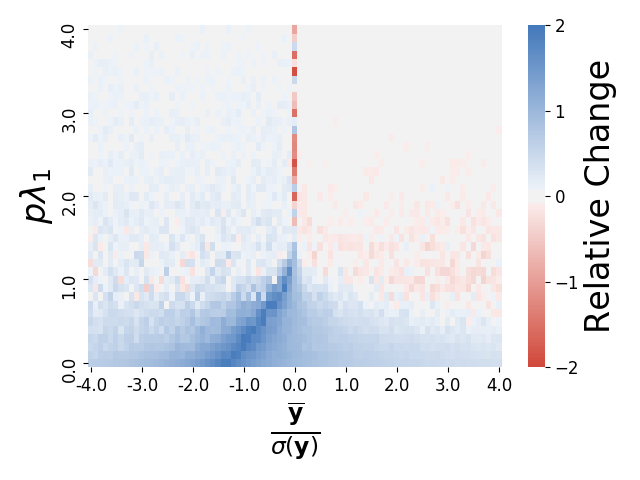} \label{fig:relative_change_erdos_renyi}} \\
\subfloat[Watts-Strogatz model] {\includegraphics[width=.4\linewidth]{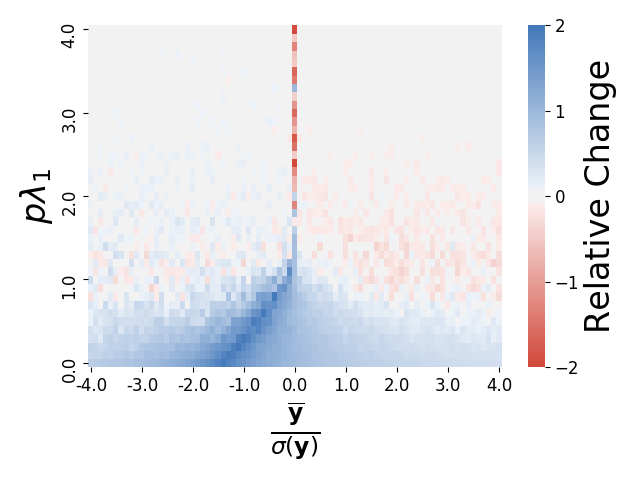} \label{fig:relative_change_watts_strogatz}}
\subfloat[All models]{\includegraphics[width=.4\linewidth]{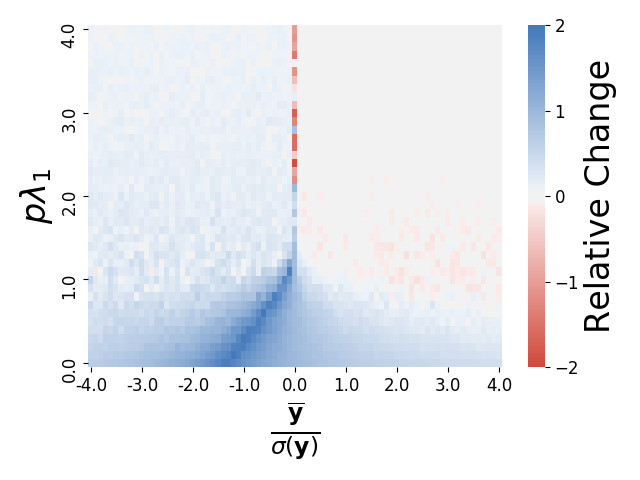} \label{fig:relative_change_synthetic}}
\caption{Performance of contextual centrality relative to other centralities on random networks. Each plot shows the relative change, computed as $\frac{a - b}{max(|a|, |b|)}$ where $a$ is CC's average payoff and $b$ is the maximum average payoff of the other centralities, for varying values of $\frac{\overline{\mathbf{y}}}{\sigma(\mathbf{y})}$ and $p \lambda_1$. The plots correspond to the results on random networks generated according to the Barabasi-Albert, Erdos-Renyi,  Watts-Strogatz models, and all models. 
}
\label{fig:perform_cc_simulated}
\end{figure*}

\paragraph{Performance of contextual centrality relative to other centralities on real-world networks}

Next, we analyze the performance of contextual centrality in achieving the cascade payoff, as defined in Eq.~(\ref{obj:cascade_payoff}), using simulations on three real-world settings, namely adoption of microfinance, adoption of the weather insurance, and political voting campaign. To compare the performance of contextual centrality against the maximum of centrality measures for each condition, we use `relative change' as before. 
We observe the network structure ($\mathbf{A}$) and adoption decisions in the  campaign for microfinance and weather insurance. In the campaign for political votes, we generate the network structure and the contribution vector from the empirical distributions. We vary the diffusion rate of $p$ in the independent cascade model to examine how it influences the performances of different centrality measures. We see that in (a) campaign for microfinance and (b) campaign for weather insurance, CC outperforms the other centralities when $p \lambda_1$ is small. While in (c) campaign for political votes, CC outperforms the other centralities for all $p \lambda_1$. The standardized average contributions of (a), (b), and (c) are  2.29, 5.27, and -2.22, respectively. This result is consistent with the results presented in Fig.~\ref{fig:perform_cc_simulated}. It shows that contextual centrality can greatly outperform other centrality measures when the standardized average contribution is negative for a wide range of $p \lambda_1$. When  standardized average contribution is positive, contextual centrality outperforms other centrality measures when the spreadability is small and achieves comparable result with other centralities as the spreadability further increases.

\begin{figure*}
\centering
\subfloat[Campaign for microfinance]{\includegraphics[width=.33\linewidth]{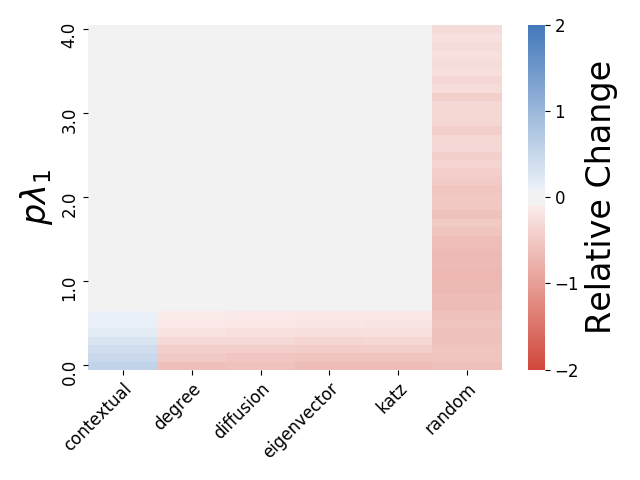} \label{fig:relative_change_microfinance}}
\subfloat[Campaign for weather insurance]{\includegraphics[width=.33\linewidth]{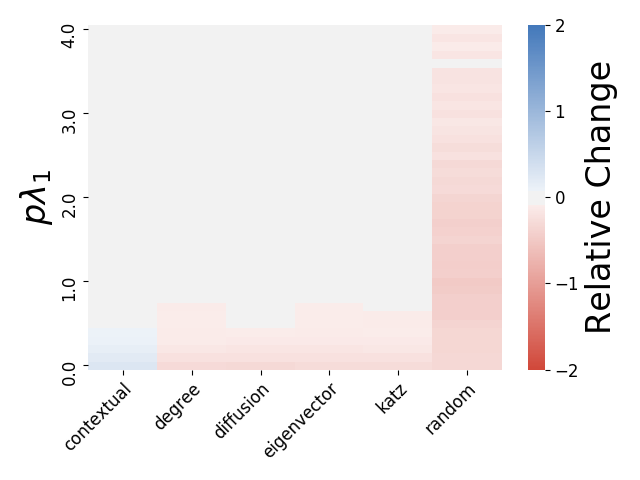} \label{fig:relative_change_insurance}}
\subfloat[Campaign for votes]{\includegraphics[width=.33\linewidth]{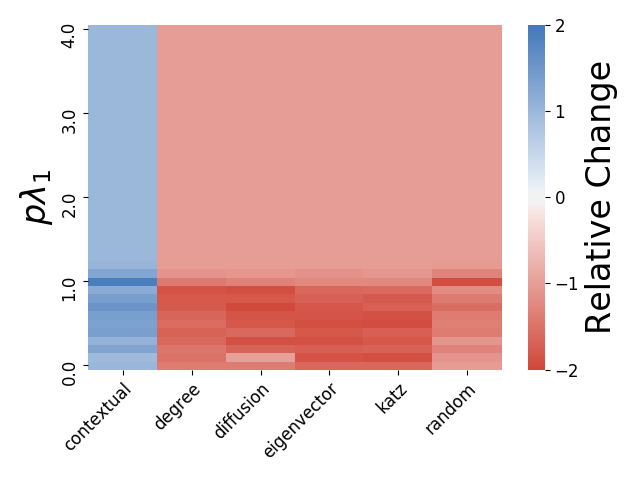} \label{fig:relative_change_political_watts_strogatz_10}}
\caption{Performance of contextual centrality relative to other centralities on real-world networks. Each plot shows the relative change for varying values of $p \lambda_1$. We compare contextual centrality with degree centrality, diffusion centrality, eigenvector centrality, Katz centrality, and random seeding. }
\label{fig:perform_cc_real}
\end{figure*}

\paragraph{Approximation of contextual centrality and the importance of primary contribution}
A negative contextual centrality score indicates that seeding with the particular node will generate a negative payoff. Therefore, we design a seeding strategy in which we seed only if the maximum of contextual centrality is nonnegative. As shown by the blue dashed and solid lines in Fig.~\ref{fig:perform_cc_large_connectivity}, the new seeding strategy, ``seed nonnegative'', performs better than always seeding the top-ranked individual. Building upon Eq.~(\ref{eq:prox_cc}), we introduce a variation of eigenvector centrality, ``eigenvector adjusted'', as the product of eigenvector centrality and the primary contribution ($\mathbf{U}_1^T \mathbf{y}$).  This variation of eigenvector centrality performs on par with contextual centrality as $p \lambda_1$ grows large as expected according to Eq.~(\ref{eq:prox_cc}). ``Eigenvector adjusted'' greatly outperforms eigenvector centrality \footnote{Another variation of eigenvector centrality is to adjust eigenvector centrality by $\overline{\mathbf{y}}$. Note that the sign of $\mathbf{U}_1^T \mathbf{y}$ does not always equal $\overline{\mathbf{y}}$. When the signs differ, seeding only when $\mathbf{U}_1^T \mathbf{y}$ is positive produces a higher cascade payoff when $p\lambda_1$ is not too large. However, as $p\lambda_1$ further increases and the diffusion saturates most of the network, the sign of $\overline{\mathbf{y}}$ predicts that of the cascade payoff. However, larger $p\lambda_1$ is not as interesting as smaller ones, which happens more frequently in real life. We present average cascade payoff comparing the two strategies when $\overline{\mathbf{y}} ( \mathbf{U}_1^T \mathbf{y}) < 0$ in the Supplementary Information.}. Comparing the strategies in Fig.~\ref{fig:perform_cc_large_connectivity}, the new strategy of accounting for the sign of the centrality measures improves the average payoffs by an order of magnitude. This pattern also highlights the importance of the primary contribution in campaign strategies. We present figures for the analogous variations of the other centralities in the Supporting Information.

\begin{figure}
\centering
{\includegraphics[width=0.5\linewidth]{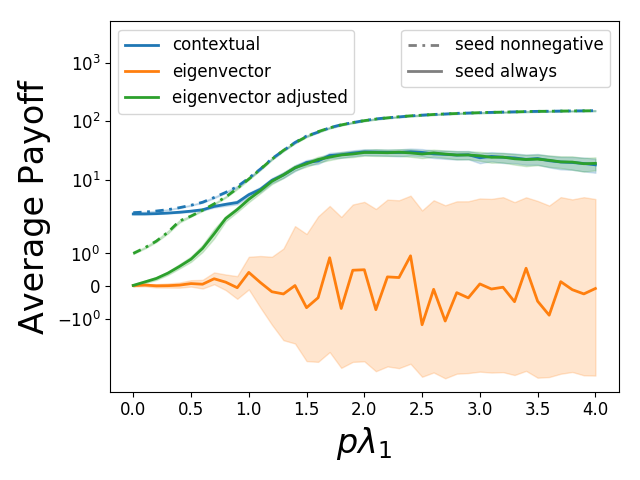} }
\caption{Average cascade payoff for variations of contextual centrality and eigenvector centrality. The x-axis is $p \lambda_1$, and the y-axis is the average payoff, with the shaded region as the 95\% confidence intervals. For `eigenvector adjusted' centrality, we multiply eigenvector centrality with the primary contribution $\mathbf{U}_1^T \mathbf{y}$. For 'seed nonnegative', we only seed if the maximum of the centrality measure is nonnegative, otherwise it is named as `seed always'. }
\label{fig:perform_cc_large_connectivity}
\end{figure}

\paragraph{Homophily and the maximum of contextual centrality} 
Homophily is a long-standing phenomenon in social networks that describes the tendency of individuals with similar characteristics to associate with one another \cite{mcpherson2001birds}. The strength of homophily is measured by the difference in the contributions of the neighbors, $\sum_{i,j}^N A_{ij} (y_i - y_j)^2$. We analyze the relationship between the strength of homophily and the approximated cascade payoff by seeding the highest-ranked node in contextual centrality. After controlling for $\frac{\overline{\mathbf{y}}}{\sigma(\mathbf{y})}$ and $p \lambda_1$, we regress the maximum of the contextual centrality on the strength of homophily of the network separately for three conditions of $\frac{\overline{\mathbf{y}}}{\sigma(\mathbf{y})}$. When the spreadability of contextual centrality is small,  stronger homophily tends to correlate with a large approximated cascade payoff across all graph types. This result shows that stronger homophily of the network predicts higher approximated cascade payoff with small spreadability. When the network is Barabasi-Albert and $\frac{\overline{\mathbf{y}}}{\sigma(\mathbf{y})} > 0$, the relationship is the strongest. As the spreadability further increases, the correlation between contextual centrality and homophily drops dramatically, and thereby we exclude the scenarios when $p \lambda_1 > 1$.

\begin{figure}
\centering
\includegraphics[width=0.5\linewidth]{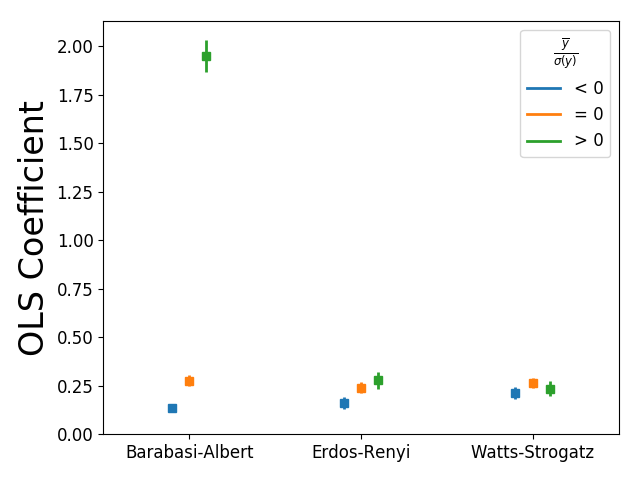}
\caption{Homophily and maximum of contextual centrality when $p\lambda_1 < 1$. We regress the maximum of contextual centrality on homophily after controlling for $\frac{\overline{\mathbf{y}}}{\sigma(\mathbf{y})}$ and $p \lambda_1$. The y-axis is the OLS coefficients of homophily (with the vertical line as the 95\% confidence interval) and the x-axis corresponds to three types of networks. We perform the analysis separately for $\frac{\overline{\mathbf{y}}}{\sigma(\mathbf{y})}$ being larger than, smaller than and equals to zero.} 
\label{fig:homo_context}
\end{figure}

\subsection*{Discussion}

Contextual centrality sheds light on the understanding of node importance in networks by emphasizing node characteristics relevant to the objective of the diffusion other than the structural topology, which is vital for a wide range of applications such as marketing or political campaigns on social networks. Notably, nodal contributions to the objective, the diffusion probability, and network topology jointly produce an effective campaign strategy. It should be obvious up to now with the thorough simulations in this study that exposing a large portion of the population in the diffusion is not always desirable. When the spreadability is small, contextual centrality effectively ranks the nodes whose local neighborhoods generate larger cascade payoffs the highest. When the spreadibility is large, the primary contribution tends to predict the sign of the approximated cascade payoff. This suggests that when the primary contribution is negative the campaigner should reduce the spreadability of the campaign to take advantage of the individuals whose local neighborhoods generate positive approximated cascade payoff in aggregation. Resorting to campaign channels with lower diffusion probability and less viral features, such as direct mail, can reduce spreadability. Moreover, as the standardized average contribution increases, the contribution vector becomes comparatively more homogeneous and comparatively less important than the network structure. Therefore, when the average contribution is positive, seeding with contextual centrality becomes similar to seeding with diffusion centrality.

Contextual centrality emphasizes the importance of incorporating node characteristics that are exogenous to the network structure and the dynamic process. More broadly, contextual centrality provides a generic framework for future studies to analyze the joint effect of network structure, nodal characteristics, and the dynamic process. Other than applications on social networks, contextual centrality can be applied to analyzing a wide range of networks, such as the biology networks (e.g., rank the importance of genes by using the size of their evolutionary family as the contribution vector\cite{martinez2003identifying}), the financial networks (e.g., rank the role of institutions in risk propagation in financial crisis with their likelihoods of failure as the contribution vector\cite{acemoglu2015systemic}), and the transportation networks (e.g., rank the importance of airports with the passengers flown per year as the contribution vector\cite{guimera2005worldwide} ).

\subsection*{Methods}

In this study, we compare contextual centrality with diffusion centrality and other widely-adopted reachability-based centrality measures - degree, eigenvector, and Katz centrality. We compute degree centrality by taking the degree of each node, normalized by $N - 1$.
We compute eigenvector centrality by taking the leading eigenvector $\mathbf{U}_1$ with unit length and nonnegative entries. We compute Katz centrality as $\sum_{t=0}^\infty (\alpha \mathbf{A})^t \mathbf{1}$, setting $\alpha$, which should be strictly less than $\lambda_1^{-1}$,  to $0.9 \cdot \lambda_1^{-1}$. We compute diffusion centrality as $\sum_{t=1}^T (p \mathbf{A})^t \mathbf{1}$. 
For both diffusion and contextual centrality, we set $T = 16$, except for the microfinance in Indian villages setting, where we set $T$ as done by Banerjee et al. \cite{banerjee2013diffusion}. 

Simulations of the diffusion process in each setting follow the independent cascade model \cite{kempe2003maximizing}. For each centrality, the highest ranked node is set to be the initial seed. We compute cascade payoff by summing up the individual contributions of all the nodes reached in the cascade. For each parameter tested in different settings, we run 100 simulations.

In the empirical analysis of microfinance in Indian villages and weather insurance in Chinese villages, we build models to predict the adoption likelihood to use as $\mathbf{y}$ in computing contextual centrality. For each setting, we use the data provided in Banerjee et al. \cite{banerjee2013diffusion} and Cai et al. \cite{cai2015social}, respectively, as inputs to a feed-forward neural network trained to predict the adoption likelihood based on the adoption decisions of first-informed individuals. For the microfinance in Indian villages, the covariates include village size, quality of access to electricity, quality of latrines, number of beds, number of rooms, the number of beds per capita, and the number of rooms per capita. For the weather insurance in Chinese villages setting, 39 of the provided characteristics are selected as inputs by choosing those for which all households had data after removing households with many missing characteristics.

For the political campaign experiment in Turkey, we use individual home and work locations to build a network and regional voting data to sample voting outcomes to use as $\mathbf{y}$. Individuals belonging to the same home neighborhood are connected according to the Watts-Strogatz model with a maximum of 10 neighbors. Same for the work neighborhoods. These two networks are superimposed to form the final network. Since we do not know the political voting preferences on an individual level, individual voting outcomes are sampled to match voting data on a regional level. Specifically, we let the actual fraction of the population that voted for the AK party in an individual's home neighborhood be the probability that individual votes for the AK party. We let $y_i = +1$ represent a vote for AK party and $y_i = -1$ represent a vote for any other party. We sample a new set of voting outcomes from the regional voting distributions for each diffusion simulation.

For the synthetic setting, we generate a new random graph for each simulation according to Barabasi-Albert, Erdos-Renyi, and Watts-Strogatz models. The size $n$ of each graph varies between 30 and 300. For Barabasi-Albert, $m$ varied between 1 and $n$. For Erdos-Renyi, $p$ varies between 0 and 1. For Watts-Strogatz, $k$ varies between $\ln{n}$ and $n$, and $p$ varies between 0 and 1. Individual contributions $\mathbf{y}$ are sampled from a normal distribution with unit standard deviation. Note that the scale of $\mathbf{y}$ does not change the rankings of contextual centrality.

\bibliographystyle{naturemag}
\bibliography{pnas-sample}

\clearpage
\section*{Supporting Information}
\tableofcontents   
\clearpage

\section{Properties of contextual centrality}

\subsection{The relationship between contextual centrality and other centrality measures}

Degree, eigenvector, Katz, diffusion, and contextual centrality can all be expressed as specific cases of a simple recurrence relation with an intuitive explanation. Roughly speaking, a node's importance in a network can be broken down into two parts: its influence on other nodes in the network through its neighbors, and its individual contribution to the cascade payoff. 

Let $\mathbf{c}_t$ be the importance (i.e., centrality) of all nodes in the network at time step $t$. One way to capture the notion of each node's influence on other nodes in the network is through $\mathbf{A} \mathbf{c}_{t-1}$, where $\mathbf{A}$ is the adjacency matrix of the network. This term effectively sums up the importance of the neighbors of each node. With this in mind, we can express $\mathbf{c}_t$ as, 
\begin{equation}
\mathbf{c}_t = \alpha \mathbf{A} \mathbf{c}_{t-1} + \boldsymbol{\beta},
\end{equation}
where $\alpha$ is a constant and $\boldsymbol{\beta}$ is the individual contribution of each node in the network. It is, of course, possible to parameterize $\alpha$, $\boldsymbol{\beta}$, or $\mathbf{A}$ by $t$ as well, but for simplicity let us assume they remain constant. Expanding this recurrence, we get

\begin{equation}
\mathbf{c}_t = (\alpha \mathbf{A})^t \mathbf{c}_0 + \sum_{i=0}^{t-1} (\alpha \mathbf{A})^i \boldsymbol{\beta}.
\end{equation}

Now if we substitute $\alpha = p$, $\boldsymbol{\beta} = \mathbf{y}$, and $\mathbf{c}_0 = \mathbf{y}$, then $\mathbf{c}_T$ is exactly equal to CC. Substitutions can be done for all the centralities discussed above and are summarized in Table \ref{table:substitutions}.

\begin{table}[h]
\caption{Centralities defined by $\mathbf{c}_t = \alpha \mathbf{A} \mathbf{c}_{t-1} + \boldsymbol{\beta}.$}
\begin{center}
\begin{tabular}{ |c|c|c|c|c| } 
 \hline
Centrality & $\alpha$ & $\boldsymbol{\beta}$ & $\mathbf{c}_0$ & $t$ \\
\hline \hline
Degree & 1 & $\mathbf{0}$ & $\mathbf{1}$ & 1 \\
\hline
Eigenvector & 1 & $\mathbf{0}$ & $\mathbf{1}$ & $\infty$ \\
\hline
Katz & $< \frac{1}{\lambda_1}$ & $\mathbf{1}$ & $\mathbf{1}$ & $\infty$ \\
\hline
Diffusion & $p$ & $\mathbf{1}$ & $\mathbf{1}$ & $T$ \\
\hline
Contextual & $p$ & $\mathbf{y}$ & $\mathbf{y}$ & $T$ \\
\hline
\end{tabular}
\end{center}
\label{table:substitutions}
\end{table}

Contextual centrality is developed upon and generalizes diffusion centrality, but there are two important differences.  First, all nodes passed through by the random walk contribute positively and homogeneously in diffusion centrality, while the main advantage of contextual centrality is allowing for the heterogeneous contributions. Second, the random walk of contextual centrality starts from the chosen seed, while that of diffusion centrality starts from the neighbors of the chosen seed. Under the condition that $\overline{\mathbf{y}}$ is positive and constant for all entries, contextual centrality inherits the nice nesting properties of diffusion centrality, which encompasses and spans the gap between degree centrality, eigenvector centrality, and Katz centrality. In particular, CC is proportional to degree centrality when $T = 1$, proportional to eigenvector centrality as $T \rightarrow \infty$ when $p\geq \lambda_1^{-1}$, and proportional to Katz centrality when $T = \infty$ and $p < \lambda_1^{-1}$. Proof can be found in Banerjee et al. \cite{banerjee2017using}.

Contextual centrality is also similar to Katz centrality, but we highlight two crucial differences. First, contextual centrality is more general in that $p$ can be larger than $\lambda_1^{-1}$ and provides essential insights into this region. Second, we allow $T$ to vary according to the specific setting while in Katz centrality the diffusion period $T$ is infinite. $T$ carries important implications. For the product that is effective in a short period, such as a coupon that will expire within a day, $T$ is relatively small compared with the diffusion of a new phone which will be on the market for much longer.

\subsection{Relationship between approximated cascade payoff and cascade payoff}

Contextual centrality aims to maximize objective~(\ref{obj:cascade_payoff_2}), which provides an approximation to cascade payoff, as in objective~(\ref{obj:cascade_payoff}), by an independent cascade model. In Fig.~\ref{fig:cascade_payoff_cc}, we analyze the Spearman's and Pearson correlation between the two concerning different spreadability. Both correlation measures decrease as spreadability increases from 0 to 1 and increase afterward.  On the bulk part, the Spearman's correlation between the two is higher than Pearson correlation and is around 0.9 or higher. Note that $p\lambda_1 = 1$ is the phase transition in network contagion with the Susceptible-Infected (SI) model and is known as the epidemics threshold\cite{chakrabarti2008epidemic}. This may explain why we see a different behavior close to $p \lambda_1 = 1$.

\begin{figure}
\centering
\includegraphics[width=.5\linewidth]{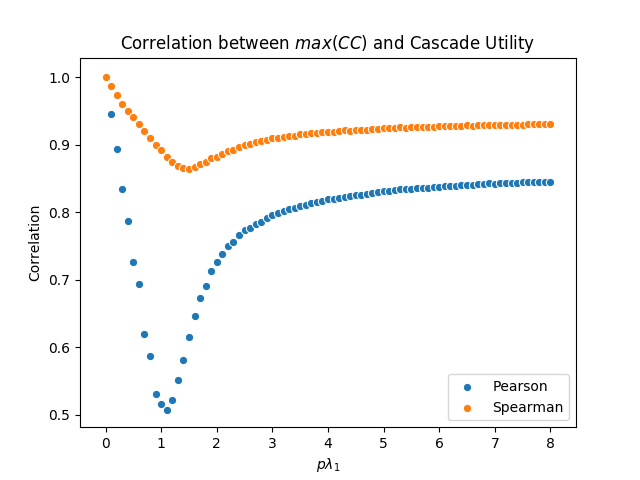}
\caption{Relationship between approximated cascade payoff and cascade payoff. The y-axis and x-axis display the correlation and the spreadability ( $p\lambda_1$) respectively. Pearson and Spearman's correlation are shown in blue and orange color respectively.} 
\label{fig:cascade_payoff_cc}
\end{figure}

\subsection{Game-theoretic interpretation of contextual centrality with local interactions} 
Ballester et al.  is the first to provide a behavioral foundation for centrality, in particular, Katz-Bonacich centrality, using a complementary linear-quadratic-form network game \cite{ballester2006s}. They found that one's network position can fully explain the Nash equilibrium in such network games. Similarly, we show that when the spreadability is smaller than one, and agents can interact for an infinite time period, their action levels can be explained by both their structural positions, as well as their marginal benefits of actions (which corresponds to nodes' contributions). 

In the setup of Ballester et al., agents choose actions optimally in response to their neighbors \cite{ballester2006s}. The quadratic functional form implies that the utility of individual $i$, ($u_i$), is quadratic in $i$'s action level ($a_i$), are dependent on $i$'s
 neighbors’ effort, and has a homogeneous marginal benefit $\alpha$ across the population, $u_i = \alpha a_i - \frac{1}{2} a_i^2 + \beta \sum_{j=1}^n A_{ij}a_i a_j$, where $\alpha$ is a scalar and $\alpha > 0$.  Taking the first-order condition, it is easy to prove that the strategy in Nash equilibrium is $\mathbf{a} = \big( \mathbf{I} - \beta \mathbf{A} \big)^{-1} \mathbf{\alpha} \propto c_\text{Katz}$, which is proportional to the Katz centrality.

In the previous setup, Ballester et al. assume that the marginal benefit is homogeneous and positive. We relax this constraint, allowing it to vary across individuals ($y_i$) with and can take on negative values. With this, suppose agent $i$ chooses an action ($a_i$) according to the following utility function, 
\begin{equation}
    u_i = a_i y_i - \frac{1}{2} a_i^2 + \beta \sum_{j=1}^n A_{ij}a_i a_j. 
\end{equation}
With this variant, the equilibrium strategy becomes, 
\begin{equation}
    \mathbf{a} = \big( \mathbf{I} - \beta \mathbf{A} \big)^{-1} \mathbf{y}. 
 \label{eq:equil_cc}
\end{equation}
Eq.~(\ref{eq:equil_cc}) has the exact same form as CC when $T \rightarrow \infty$, $ \beta \lambda_1 < 1$ and $\beta = p$. Hence, we see that contextual centrality approximates agents' equilibrium actions with heterogeneous marginal utilities in this condition. 

\subsection{Bounds and distribution of contextual centrality in terms of spreadability}

In this section, we present the upper bound for maximum possible contextual centrality. When $p\lambda_1$ is larger than $1$,  CC approaches infinity as $T$ grows. On the other hand, when $p\lambda_1 < 1$, CC is finite for $T=\infty$, which can be understood as a lack of virality, expressed in a fizzling out of the diffusion process with time. In fact, the specific value of $p \lambda_1$ can be used to bound the maximum possible CC given the norm of the score vector $\textbf{y}$.


\begin{proposition}
\label{ccineq}
\begin{equation*}
max(\text{CC}(\mathbf{A}, p, T, \mathbf{y}))\le ||\text{CC}(\mathbf{A}, p,T, \mathbf{y})||
\end{equation*}
\begin{equation*}
 \le \frac{1-(p\lambda_1)^{T+1}}{1-p\lambda_1}||\mathbf{y}||
\end{equation*}
If, in addition, $p\lambda_1<1$, then this is further bounded by $\frac{1}{1 - p\lambda_1}||\mathbf{y}||$.
\end{proposition}

\begin{proof}
The first inequality, $max(CC(\mathbf{A}, p, T, \mathbf{y}))\le ||CC(\mathbf{A}, p, T, \mathbf{y})||$, is clear.

Next we use the matrix norm $||\mathbf{A}||:=sup\{||\mathbf{A}x||/||x||:x\ne 0\}$, which by definition satisfies $||\mathbf{A}x||\le ||\mathbf{A}||\cdot ||x||$ for all $x$, and which coincides with spectral radius $\rho(\mathbf{A})$ for symmetric matrices. Since, for us, $\rho(\mathbf{A})=\lambda_1$, we have
\begin{equation*}
||CC(\mathbf{A}, p, T, \mathbf{y})||=||(\sum_{t=0}^T(p\mathbf{A})^t)\mathbf{y}||\le ||(\sum_{t=0}^T(p\mathbf{A})^t)||\cdot||\mathbf{y}||
\end{equation*}
\begin{equation*}
\le \sum_{t=0}^T||(p\mathbf{A})^t||\cdot||\mathbf{y}||=\sum_{t=0}^T(p\lambda_1)^t\cdot||\mathbf{y}||\le \frac{1-(p\lambda_1)^{T+1}}{1-p\lambda_1}||\mathbf{y}||
\end{equation*}
which, if $p\lambda_1<1$, can be further bounded by $ \frac{1}{1 - p\lambda_1}||\mathbf{y}||$
\end{proof}

While the above result bounds contextual centrality from above, the actual value of CC is highly variable, depending on the structure of the graph and the distribution of the score vector among its nodes. For a discussion of expected CC among random networks, see the Erdos-Reyni section below. Next, we discuss the behavior of contextual centrality when $\mathbf{y}$ is variable.

\subsection{Robustness of contextual centrality in response to perturbations in $\mathbf{y}$}

As discussed in the main body of the paper, in the real-world data, node characteristics can be noisy, stochastic and biased. Therefore, it is essential to analyze the robustness of contextual centrality in response to small perturbations in $\mathbf{y}$. We first perform a sensitivity analysis, studying bounds on the error in contextual centrality in terms of noise in $\mathbf{y}$, and then study contextual centrality as a random variable assuming a multivariate normal model of $\mathbf{y}$.

\paragraph{Sensitivity Analysis}
We let the observed (or estimated) score vector be $\hat{\mathbf{y}}$ and let $\mathbf{y}$ be the true score vector. The errors in the score vector are given by the vector $\Delta \mathbf{y}: = {\mathbf{y}} -  \hat{\mathbf{y}} $ and similarly $\Delta \text{CC}:=\text{CC}(\mathbf{A}, p,T,\hat{\mathbf{y}})-\text{CC}(\mathbf{A}, p,T,\mathbf{y})$ is the error between the CC computed from observed and actual data. 

We have the following bound on $||\Delta \text{CC}||$, which follows directly from Proposition \ref{ccineq} and the fact that CC is linear with respect to the score vector $\mathbf{y}$.

\begin{corollary}
\begin{equation*}
||\Delta CC||
\le \frac{1-(p\lambda_1)^{T+1}}{1-p\lambda_1}||\Delta \mathbf{y}||
\end{equation*}
If, in addition, $p\lambda_1<1$, then this is further bounded by $\frac{1}{1 - p\lambda_1}||\Delta \mathbf{y}||$.
\end{corollary}

This shows that, when $p\lambda_1<1$, then as long as the error in $\mathbf{y}$ is sufficiently small, the error in CC will be small as well. However, the larger $p\lambda_1$ is, the more a small error in $\mathbf{y}$ can become amplified as an error in CC.

Next we focus on the case that $p\lambda_1>1$. In this case, we have shown in the main body of the paper that for large T, contextual centrality is well-approximated by $(U_1^T\mathbf{y})U_1$, where $U_1$ is the eigenvector with largest eigenvalue. Thus in this case, the primary contribution $U_1^T\mathbf{y}$ is an essential quantity whose sign roughly determines the relative ranking of contextual centrality. Hence, we analyze its sensitivity to noise in $\mathbf{y}$. The error in primary contribution is simply $\mathbf{U}_1^T \Delta \mathbf{y}$, whose magnitude is bounded by $||\Delta \mathbf{y}||$. Thus if $\Delta \mathbf{y}$ is small enough so that $||\Delta \mathbf{y}||< \mathbf{U}_1^T\hat{\mathbf{y}}$, this perturbation will not affect the sign of the primary contribution, so the relative ranking in CC will tend to stay fixed. Otherwise, the relative ranking is at risk of flipping.

\paragraph{Contextual centrality as a random variable}
Next, to study the impact of stochasticity in $\mathbf{y}$, we suppose  that $\mathbf{y}$ is a multivariate random variable with mean vector $\hat{\mathbf{y}}$ and covariance matrix $\Sigma$. Let $\mathbf{B}:=\sum_{t=0}^T (p\mathbf{A})^t$. Since $\text{CC}(\mathbf{A}, p,T,\mathbf{y})=\mathbf{B} \cdot \mathbf{y}$ is a linear transformation of the multivariate normal variable $\mathbf{y}$, it is also a multivariate normal variable, with mean $\mathbf{B}\hat{\mathbf{y}}=\text{CC}(\mathbf{A}, p,T, \hat{\mathbf{y}})$ and covariance matrix $\mathbf{B} \boldsymbol{\Sigma} \mathbf{B}$.

To simplify, consider the case that $\boldsymbol{\Sigma}=\sigma^2 \mathbf{I}$, that is, the $y_i$ are uncorrelated and have the same standard deviation $\sigma$. Then the covariance matrix of $\text{CC}(\mathbf{A}, p,T,\mathbf{y})$ is $\sigma^2\mathbf{B}^2$.

That is, we have
\begin{equation*}
\text{Cov}(\text{CC}(\mathbf{A}, p,T,\mathbf{y})_i,\text{CC}(\mathbf{A}, p,T,\mathbf{y})_j)=\sigma^2 (\mathbf{B}e_i)\cdot (\mathbf{B}e_j), 
\end{equation*}
where $e_i$ are the standard basis vectors.

In particular, the coefficients of CC may be positively correlated even when those of $\mathbf{y}$ are uncorrelated, and their standard deviations are given by 
\begin{equation*}
\sigma(\text{CC}(\mathbf{A}, p,T,\mathbf{y})_i)=\sigma ||\mathbf{B}e_i||
\end{equation*}
Note that, by definition of $\mathbf{B}$, $\mathbf{B}e_i=\text{CC}(\mathbf{A}, p,T,e_i)$, whose $j$th coefficients represents the expected number of times node $i$ is reached by the diffusion process, if seeded at node $j$.

By Proposition \ref{ccineq}, we have the bound $\sigma(\text{CC}(\mathbf{A}, p,T,\mathbf{y})_i)=\sigma ||\mathbf{B}e_i||\le \frac{\sigma }{1-p\lambda_1}$ if $p\lambda_1<1$.

\subsection{Differences between contextual centrality and centrality measures developed on weighted networks}
There have been some studies which generalize centrality measures to weighted or signed networks. They focus on settings where edge weights represent the strength or the trustiness (a friend or a foe) of the social relationships. The network information captured by these centrality measures can be regarded as a special case of a weighted version of contextual centrality, where $\mathbf{A}$ is a weighted matrix, $p = 1$ and $\overline{\mathbf{y}} = \mathbf{1}$. Weights on network links emphasize social relationships but do not capture the heterogeneous contributions of the nodes - exogenous to the network structure - directly to the cascade payoff. Weighted links and weighted nodes characterize different network dynamics and diffusion objective. Let us provide a simple illustrative example to explain the differences better. Imagine a network with two disconnected communities, where one component consists of positive links and the other consists of negative links. Centrality developed on the weighted or signed network will rank the most-connected node as the top in the community with positive edges (i.e., individuals who all trust one another). However, for a particular marketing campaign, if all individuals in the positive community do not like the product, seeding any individuals in the positive community will hurt the campaign. 

For readers' reference, we provide an overview of centrality measures on weighted and signed networks. There are two main strands of work in this literature. First, some studies define new notions of the shortest path that take the weights of the links into account. There are multiple types of modifications: (1) take the inverse of the tie strengths as the shortest path lengths \cite{newman2001scientific,brandes2008variants}, (2) using a tuning parameter to trade-off tie strengths and the number of ties \cite{opsahl2010node}, (3) adding a temporal aspect to links to minimize the temporal latency \cite{williams2016spatio}. With these new notions, researchers extend existing path-based centrality measures. \cite{newman2001scientific} and \cite{brandes2008variants} extended closeness centrality and betweenness centrality to define the shortest path algorithm to be the least costly path with cost depending solely on tie weights. 
Opshal proposes a centrality measure with a generalized degree and shorted paths computation by adding a tuning parameter on tie strengths \cite{opsahl2010node}. Another strand of studies focused on the flow and diffusion processes. 
Kunegis developed a signed centrality measure using the left eigenvector of the signed network as a generalization of the eigenvector centrality with weighted edges \cite{kunegis2009slashdot}. Other studies develop algorithm-based ranking methods, extending PageRank or HITS. Shahriari computed the difference between the scores using PageRank or HITS algorithms for networks consisting of positive and negative links, respectively, as the new measure \cite{shahriari2014ranking}

\subsection{Theoretical results of contextual centrality for Erdos-Renyi networks} In the case that $A$ corresponds to an Erdos-Renyi graph $G(n,q)$, we have further theoretical results, in line with the results of Banerjee et al.\cite{banerjee2014gossip} on diffusion centrality. As is standard for Erdos-Renyi graphs, we assume each edge has independent probability $q$ of being present in the graph, where $q$ is a function of $n$, the number of nodes. Assume that $qn$ grows such that $log(n)\le qn\le \sqrt{n}$. We also assume that $T$ and $p$ are functions of $n$, and let $\textbf{y}$ denote the vector (depending on $n$) consisting of $y_1,\dots,y_n$ for some infinite sequence $\{y_i\}$. We suppress all dependency on $n$ for ease of notation. We further assume that the mean $\overline{\textbf{y}}$ has a limit $\overline{\mathbf{y}}$ as $n$ approaches infinity, which is reasonable by the law of large numbers if the $y_i$ are sampled from a random variable. With this background, we study the expected behavior of $E(\text{CC}(\mathbf{A}, p, T, \mathbf{y}))$.

Given two functions $f(n),g(n)$, we will say that $f$ approaches $g$ as $n$ approaches infinity, if $lim_{n\to \infty} \frac{f(n)}{g(n)}=1$. Then we have the following result.

\begin{theorem}
\label{expccdc}
Suppose $T=o(qn)$ and $\log(n)\le qn \le \sqrt{n}$. Then we can decompose $E(CC(\mathbf{A}, p, T, \mathbf{y})_i)=\overline{\mathbf{y}}E_1+y_iE_2$, where $E_1$ and $E_2$ are functions of $n,p,q,T$ but do not depend on $\mathbf{y}$ or $i$, such that

a)$E_1$ approaches $\frac{1-(npq)^{T+1}}{1-(npq)}$.

b)$E_2=o(E_1)$.

c)If $\overline{\mathbf{y}}\ne 0$, then $E(\text{CC}(\mathbf{A}, p, T, \mathbf{y}))$ approaches $\overline{\mathbf{y}}E(\text{DC}(\mathbf{A}, p, T, \mathbf{y}))$, where DC is diffusion centrality.
\end{theorem}
In other words, if $\overline{\mathbf{y}}\ne 0$, then the term $E_1$ dominates, so the expected CC is uniform all nodes (in the limit as $n$ approaches infinity). Moreover, $\overline{\mathbf{y}}$ measures the magnitude of the diffusion as compared to DC, and the sign of $\overline{\mathbf{y}}$ determines the expected sign of CC. In contrast, if $\overline{\mathbf{y}}=0$, then CC equals $E_2$ so, on expectation, CC correlates perfectly with $\mathbf{y}$ itself. We note that in practice it is not likely for $\overline{\mathbf{y}}$ to equal $0$. However, if $\overline{\mathbf{y}}$ is close to $0$ and $n$ is not too large, then the term $E_2$ could still be significant, indicating that the expected CC will be correlated with the nodal evaluation vector $\mathbf{y}$.

This result can also be related to the tradeoff in Eq.~(\ref{eq:std_mean_tradeoff}). As implied by the Theorem, as long as $\bar{y}\ne 0$, then expected CC approaches $\bar{Y}E_1$, which in turn approaches $\bar{y}E(DC)$ as $n$ approaches infinity. Thus the second term of the tradeoff in Eq.~(\ref{eq:std_mean_tradeoff}) dominates, on expectation.

We also note that careful analysis will show that $E_2>0$, but that is beyond the scope of the present paper.

\begin{theorem}
\label{thresh}
If $p\lambda_1\ge (1+\epsilon)$ for some $\epsilon>0$, then $T=\frac{log(n)}{log(npq)}$ is a threshold for viral spread if $\overline{\mathbf{y}}\ne 0$, in the sense that 

a)If $T\le (1-\epsilon)\frac{log(n)}{log(npq)}$ for some $\epsilon>0$, then $E(CC(\mathbf{A}, p,T, \mathbf{y})_i)=o(n)$ for all $i$.

b)If, on the other hand, $T\ge (1+\epsilon)(\frac{log(pnq)}{log(n)})$, then $E(CC(\mathbf{A}, p, T, \mathbf{y})_i)=\Omega(n)$ for all $i$.
\end{theorem}
Note that the threshold $T=\frac{log(n)}{log(npq)}$ given above is equal to $\frac{log(n)}{log(pE(\lambda_1))}$, since $E(\lambda_1)=nq$. We also note that the expected diameter of the Erdos-Reyni graph is $\frac{log(n)}{log(nq)}$, which is strictly smaller than the threshold given above.

To prove these theorems, we analyze $E(\mathbf{A}^t)$ for any $t$. Note that $E(\mathbf{A}^t)_{ij}$ is the weighted sum of all paths of length $t$ from $i$ to $j$, with each path $\pi$ weighted by $q^{d(\pi)}$, where $d(\pi)$ is the number of distinct edges along the path $\pi$. Note that by symmetry, the off-diagonal entries of $E(\mathbf{A}^t)$ are all the same, as are its diagonal entries; however, the diagonal entries are not necessarily equal to the off-diagonal ones.

We first prove the following lemma to aid our analysis.

\begin{lemma}
\label{helper}
Let $i,j,k$ be distinct numbers ranging from $1$ to $n$. Let $Z_{ij,k}(t)$ be the subset of paths of length $t$ from $i$ to $j$ which visit vertex $k$ at some point. Let $z_{ij,k}(t)$ be its weighted sum $\sum_{\pi\in Z_{ij,k}(t)} q^{d(\pi)}$. Then $z_{ij,k}\le \frac{t-1}{n-2}E(\mathbf{A}^t_{ij})$.
\end{lemma}

\begin{proof} There are $(t-1)$ possible indices to place the vertex $k$. For each fixed index, the weighted sum of all paths with vertex $k$ at that index is $\le \frac{1}{n-2}E(\mathbf{A}^t_{ij})$, which follows by symmetry with respect to the $n-2$ possible choices of $k$. Combining these factors yields the desired bound.
\end{proof}

We now move on to the estimates of $E(\mathbf{A}^t_{ij})$.

\begin{lemma}
\label{eccestimate}
For the purposes of this lemma assume that  $\frac{t}{nq}\le r<\frac{1}{4}$ for some $r$. Then we have

a) $(1-2r)\frac{(nq)^t}{n}\le E(\mathbf{A}^t)_{ij}\le (\frac{1}{1-4r})\frac{(nq)^t}{n}$, if $i\ne j$ or if $i=j$ and $t$ is odd.

b) $(1-2r)\frac{(nq)^t}{n}\le E(\mathbf{A}^t)_{ii}\le (\frac{1}{1-4r})(\frac{(nq)^t}{n}+(2nq)^{t/2})$ if $t$ is even.
\end{lemma}
\begin{proof}
Let us represent a path by the sequence of the vertices it visits. A path $\pi$ of length $t$ from $i$ to $j$ is represented as $iv_1v_2\cdots v_{t-1}j$, where $i$ and $j$ will also be labeled $v_0$ and $v_t$, respectively.
        
We begin by proving the lower bounds. We have $E(\mathbf{A}_{ij}^t)\ge (n-2)^{t-1}q^t$. Indeed, there are more than $(n-2)^{t-1}$ legitimate paths in $X_{ij}(t)$ (under the constraint of no self-edges), and each one has at most $t$ distinct edges. Now, $(n-2)^{t-1}\ge n^{t-1}-2t(n)^{t-2}= n^{t-1}(1-\frac{2t}{n})\ge (1-2r)n^{t-1}$ since $\frac{t}{n}\le \frac{t}{qn}\le r$. 
        
Next, we calculate the upper bounds. Suppose that $t\ge 1$. Let $Y_{ij}(t)\subset X_{ij}(t)$ consist of those paths in which edges are never repeated immediately, that is, $v_l\ne v_{l+2}$ for any index $l$. Let $y_{ij}(t)=\sum_{\pi\in Y_{ij}(t)} q^{d(\pi)}$ be its weighted sum. We further partition $Y_{ij}(t)$ as follows. For each $k=1,\dots,(t-1)$, let $Y_{ij,k}(t)\subset Y_{ij}(t)$ be the subset of those paths for which $k$ is the smallest index such that the edge $v_{k-1}v_k$ is not revisited later in the path, and $v_k\ne j$. Then $Y_{ij}(t)=\bigsqcup_{k=0}^{t-1} Y_{ij,k}(t)$. Let $y_{ij,k}(t)$ be the weighted sum of $Y_{ij,k}(t)$. Also, let $y_{diff}(t)$ and $y_{same}(t)$ denote the values of $y_{ij}(t)$ in the cases $i\ne j$ and $i=j$, respectively.

We will use the following properties for paths $\pi\in Y_{ij,k}(t)$. Given $\pi$, let $\pi'\in Y_{v_kj}(t-k)$ be the truncated path $v_k,\dots,v_{t-1},j$. We note that $\pi$ has at least one edge that $\pi'$ does not, namely $v_{k-1}v_k$, by definition of $k$. Thus $d(\pi)\ge d(\pi')+1$. Furthermore, we note that every node $v_1,\dots,v_{k-1}$ must be present in $\pi'$. Indeed, for each such vertex $v$, consider the greatest index $l$ such that $v_l=v$. If $l<k$, then, by definition of $k$, that means either $v_l=j$, in which case it appears in $\pi'$, or the edge $v_{l-1}v_l$ reappears later in the path. By assumption that $\pi\in Y_{ij}(t)$, this edge cannot be repeated immediately; hence $v=v_l$ itself must reappear later, contradicting the description of the index $l$. So, $l\ge k$, that is, $v$ indeed appears in $\pi$'. 

These observations imply the following bound: 
\begin{equation}
\label{looseboundyij}
 y_{ij,k}(t)\le t^{k-1}nqy_{diff}(t-k)
\end{equation}

Indeed, to specify a path in $Y_{ij,k}(t)$, we first choose $v_k$ from among $\le n$ possibilities. Then we choose the truncated path $\pi'$ as described above from $Y_{v_kj}(t-k)$, whose weighted sum is $y_{v_k,j}(t-k)$. Then we choose the $k-1$ vertices $v_1,\dots,v_{k-1}$. Each of them is repeated in $\pi'$, hence may be chosen from among the $\le t$ vertices of $\pi'$. Finally, since $d(\pi)\ge d(\pi')+1$, we introduce the additional factor of $q$. 

Now we focus on the case that $i\ne j$. 

If $k\ge 1$, we can improve our bound further. Notice that, since $k>1$, the starting vertex $i$ must appear in the path $\pi'$. So, either $i=v_k$, or $i\ne v_k$. In the former case, we can eliminate a factor of $n$ from (\ref{looseboundyij}), and in the latter case, we can introduce a factor of $\frac{t}{n}$ into (\ref{looseboundyij}), by Lemma \ref{helper} (Note the Lemma applies since $i,j,$ and $v_k$ are assumed distinct). We thus obtain the tighter bound
\begin{equation}
\label{tightboundyij}
y_{ij,k}(t)\le t^{k}q\cdot y_{diff}(t-k)
\end{equation}

Now we can prove by induction that $y_{diff}(t)\le (nq+2)^t$. Indeed, under this inductive hypothesis, the above bounds yield
\begin{equation*}
y_{ij,1}(t)\le nq(nq+2)^{t-1}
\end{equation*}
and
\begin{equation*}
\sum_{k=2}^{t-1} y_{ij,k}(t)\le \sum_{k=2}^{\infty} t^qk(nq+2)^{t-k}\le t^2q(nq+2)^{t-2}\frac{1}{1-r}
\end{equation*}
\begin{equation*}
\le 2(nq+2)^{t-1}
\end{equation*}
Where we used the fact that $(t^2)\le (nq)^2\le n$ and $\frac{1}{1-r}\le 2$.
Combining these bounds together we obtain, as desired, that 
\begin{equation*}
y_{ij}(t)=y_{ij,1}(t)+\sum_{k=2}^{t-1}y_{ij,2}(t)\le (nq+2)(nq+2)^{t-1}=(nq+2)^t
\end{equation*}
        
Next, we plug in this bound for $y_{diff}(t)$ into $(\ref{looseboundyij})$, to obtain a bound for $y_{ij}(t)$ (even if $i=j$). We have 
\begin{equation*}
y_{ij}\le \sum_{k=1}^{\infty}t^{k-1}(nq+2)^{t-k+1}\le \frac{1}{1-r}(nq+2)^t
\end{equation*}

Now, it is convenient to further bound $(nq+2)^t\le \frac{1}{1-2r}(nq)^t$. Indeed, $(nq+2)^t=\sum_{k=0}^t \binom{t}{k} 2^k n^{t-k}\le \sum_{k=0}^\infty (2t)^k(nq)^{t-k}$, which is a geometric series with ratio $\frac{2t}{nq}\le 2r$, bounded by $\frac{1}{1-2r}(nq)^t$.

Hence, we obtain the following bound on $y_{ij}(t)$:
\begin{equation}
\label{yineq}
y_{ij}(t)\le \frac{1}{(1-r)(1-2r)}(nq)^t\le \frac{1}{1-3r}(nq)^t
\end{equation}
We emphasize that this inequality holds only if $t\ge 1$.        
Finally, we extend our analysis from the $Y_{ij}(t)$ to all paths. Any arbitrary path from $i$ to $j$ of length $t$ may be obtained by starting with a path in $Y_{ij}(t-2m)$, for some $0\le m\le t/2$ and performing a sequence of $m$ insertions, replacing a vertex $v$ with  $vwv$ instead, for some vertex $w$. We obtain the bound
\begin{equation}
\label{ccsum}
E(\mathbf{A}^t_{ij})\le \sum_{m=0}^{\lfloor t/2 \rfloor} y_{ij}(t-2m)\cdot (2nq)^m
\end{equation}
Indeed, for each insertion operation, there are two cases: either the inserted vertex $w$ is already present in the path, so it can be chosen from among $\le t$ vertices; or it is not already present, in which case it can be chosen from among $\le n$ vertices and introduces a new edge, for an additional factor of $q$. Combining the two possibilities, each insertion operation introduces a factor of $(t+nq)\le 2nq$.
        
To evaluate this sum, we need to consider the two cases outlined in the statement of this lemma.
        
a) Suppose that either $i\ne j$, or $i=j$ and $t$ is odd. In this case, note that the bound (\ref{yineq}) can be applied to each $y_{ij}(t-2m)$, since if $t$ is odd, then $t-2m\ge 1$; and if $i\ne j$, we have $y_{ij}(0)=0$ regardless. Combining these bounds with (\ref{ccsum}), we obtain
\begin{equation*}
E(\mathbf{A}^t_{ij})\le \frac{1}{(1-3r)}\sum_{m=0}^{\infty}\frac{1}{n}2^m(nq)^{t-m}\le \frac{1}{1-4r}\frac{1}{n}(nq)^t
\end{equation*}
by a geometric series with ratio $\frac{2}{nq}<r$. This completes the proof of part a) of the lemma.
        
b) Now suppose that $i=j$ and $t$ is even. The sum in (\ref{ccsum}) can be analyzed in the same way as in a), but with an extra term of $(2qn)^{t/2}$ corresponding to the case $m=\frac{t}{2}$.
\end{proof}

We are now ready to prove Theorem \ref{expccdc}.

\begin{proof} By definition, $E( \text{CC}(\mathbf{A}, p, T, \mathbf{y}))=E(\sum_{t=0}^T p^t\mathbf{A}^t\mathbf{y})$. By linearity of expectation, this equals $\sum_{t=0}^T p^tE(\mathbf{A}^t)\mathbf{y}$.  Now, for each $t$ and each $i$, we have $(E(\mathbf{A}^t)\mathbf{y})_i=\sum_{j=0}^n y_j p^tE(\mathbf{A}^t_{ij})$. By separating the terms with $i=j$ from the terms with $i\ne j$, this equals $n\overline{\mathbf{y}} p^tE(\mathbf{A}^t_\text{diff}) +y_ip^t\cdot (E(\mathbf{A}^t_\text{same})-E(\mathbf{A}^t_\text{diff}))$, so we can write
\begin{equation*}
E(\text{CC}(\mathbf{A}, p,T, \mathbf{y})_i)=\overline{\mathbf{y}}E_{1}+y_i E_{2}
\end{equation*}
where
\begin{equation*}
E_{1}=n\sum_{t=0}^T p^tE(\mathbf{A}^t_\text{diff})
\end{equation*}
and
\begin{equation*}
E_{2}=\sum_{t=0}^T p^t\cdot (E(\mathbf{A}^t_\text{same})-E(\mathbf{A}^t_\text{diff}))
\end{equation*}

a) By Lemma \ref{eccestimate}, we know that $E_{1}$ can be bounded
\begin{equation*}
(1-2r)\sum_{t=0}^T (npq)^t\le E_{1}\le \frac{1}{1-4r} \sum_{t=0}^T (npq)^t
\end{equation*}
where $r=\frac{T}{nq}$. Since we assume this ratio approaches $0$, these bounds imply that indeed $E_{1}$ approaches $\sum_{t=0}^T (npq)^t=\frac{(npq)^{T+1}}{1-npq}$ as $n$ tends to infinity.

b) Next, we show that $E_{2}=o(E_{1})$. Indeed, we again use Lemma \ref{eccestimate}. We have
\begin{equation*}
|E_{2}|\le \sum_{t=0}^T p^t\cdot (E(\mathbf{A}^t_\text{same})+E(\mathbf{A}^t_\text{diff}))
\end{equation*}
\begin{equation*}
\le \frac{1}{1-4r}(\sum_{t=0}^T (p^t(2nq)^{t/2})+\frac{(npq)^t}{n})
\end{equation*}
so the result follows since both terms $p^t(2nq)^{t/2}$ and $\frac{(npq)^t}{n}$ are lower-order than $(npq)^t$.

c) Diffusion centrality is a special case of contextual  centrality in which $\mathbf{y}=\mathbf{1}$, which has mean $1$. The result follows by part a), together with the fact that $E_1$ dominates over $E_2$ whenever $\overline{\mathbf{y}}\ne 0$ by part b).
\end{proof}

Next, we prove Theorem \ref{thresh}.

\begin{proof}
Suppose $\overline{\mathbf{y}}\ne 0$ and $pE(\lambda_1)\ge (1+\epsilon)$ and that. For Erdos-Renyi graphs, $E(\lambda_1)=nq$, so $pnq\ge (1+\epsilon)$. In this case, it follows from Theorem \ref{expccdc} that $E(CC)_i$ approaches $\overline{\mathbf{y}}(pnq)^{T}$. If $T\le (1-\epsilon)\frac{log(n)}{log(npq)}$ for some $\epsilon>0$, then $log(|E(CC)_i|)\le C+log(|(pnq)^T|)$ for some constant $C$, which equals $C+Tlog(pnq)\le C+(1-\epsilon)log(n)$ so $T=O(n^{1-\epsilon})=o(n)$. The other direction follows similarly.
\end{proof}

\pagebreak

\section{Supplementary figures for empirical analysis}

\subsection{Predictive power of contextual centrality in eventual adoptions}

Here we include the supplementary results to examine the robustness of the predictive power of contextual centrality in the eventual adoption outcomes similar to Fig.~\ref{fig:r2}. We extend the linear regression models to (1) without controlling for village size, and (2) with additional controls.

\begin{figure*}[ht!]
\centering
\subfloat[Adoption of microfinance in Indian villages]{\includegraphics[width=0.5\linewidth]{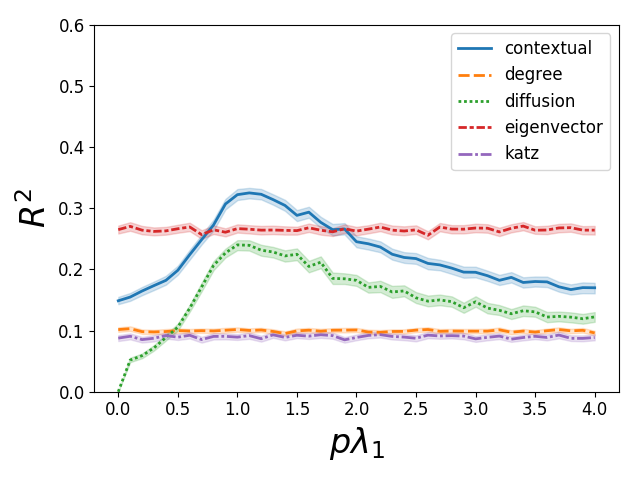}}
\subfloat[Adoption of weather insurance in Chinese villages]{\includegraphics[width=0.5\linewidth]{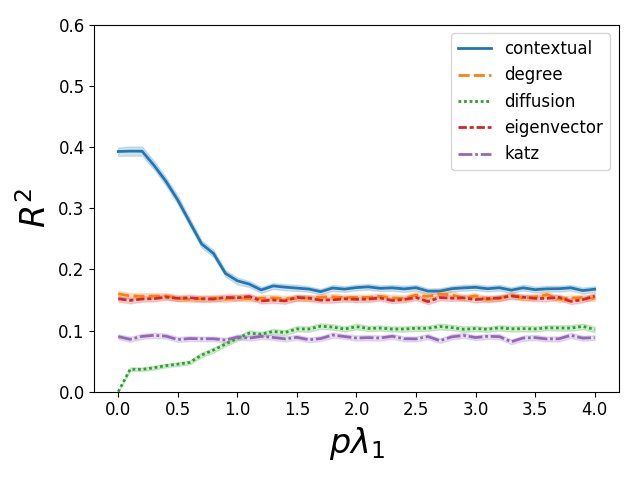}}
\caption{Predictive power of contextual centrality comparing with other centrality measures without any controls. }
\label{fig:r_squared_no_controls}
\end{figure*}

\begin{figure*}[ht!]
\centering
\subfloat[Adoption of microfinance in Indian villages]{\includegraphics[width=0.5\linewidth]{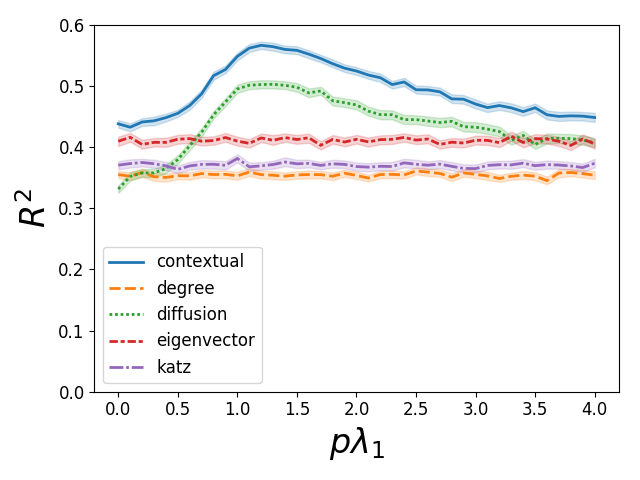}}
\subfloat[Adoption of weather insurance in Chinese villages]{\includegraphics[width=0.5\linewidth]{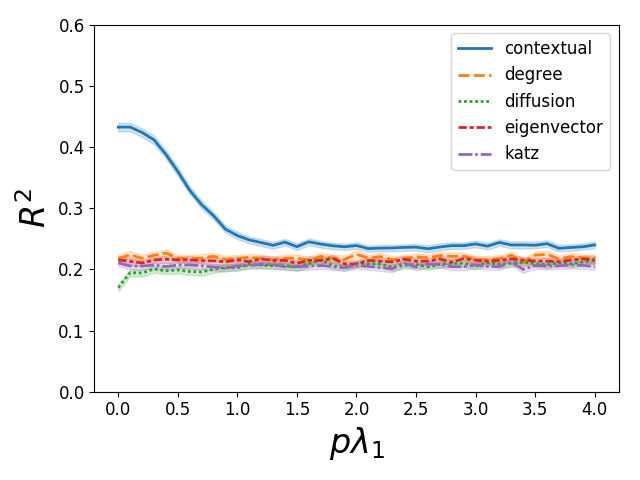}}
\label{fig:r_squared_all_controls}
\caption{ Predictive power of contextual centrality with additional controls. For (a), we use village size, savings, self-help group participation, fraction of general caste members, and the fraction of village that is first-informed as done in \cite{banerjee2013diffusion}. For (b), we use village size, number of first-informed households, and fraction of village that is first-informed.
}
\end{figure*}

\clearpage

\subsection{Performance relative to other centralities on random networks}

Here we show supplementary results corresponding to Fig.~\ref{fig:perform_cc_simulated} with figures for degree, eigenvector, Katz and diffusion centrality.

\begin{figure*}[ht!]
\centering
\subfloat[Barabasi-Albert model]{\includegraphics[width=.4\linewidth]{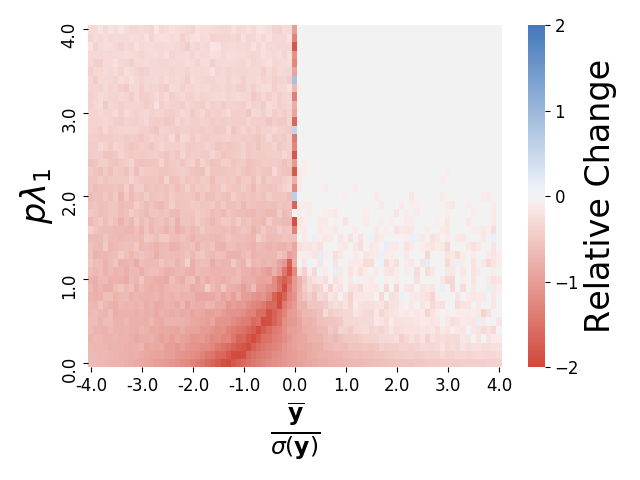} \label{fig:relative_change_degree_barabasi_albert}}
\subfloat[Erdos-Renyi model]{\includegraphics[width=.4\linewidth]{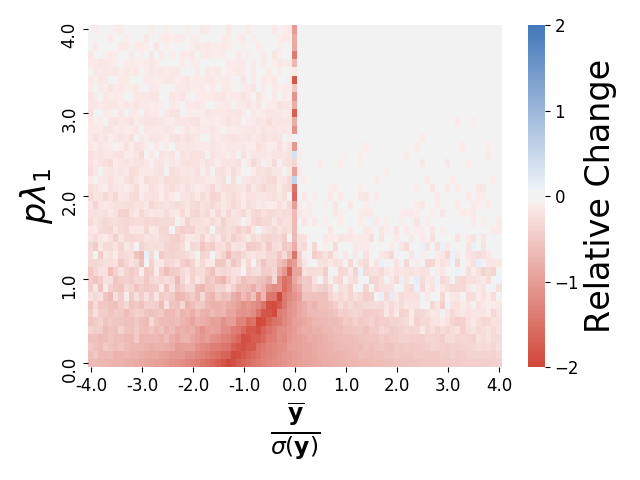} \label{fig:relative_change_degree_erdos_renyi}} \\
\subfloat[Watts-Strogatz model] {\includegraphics[width=.4\linewidth]{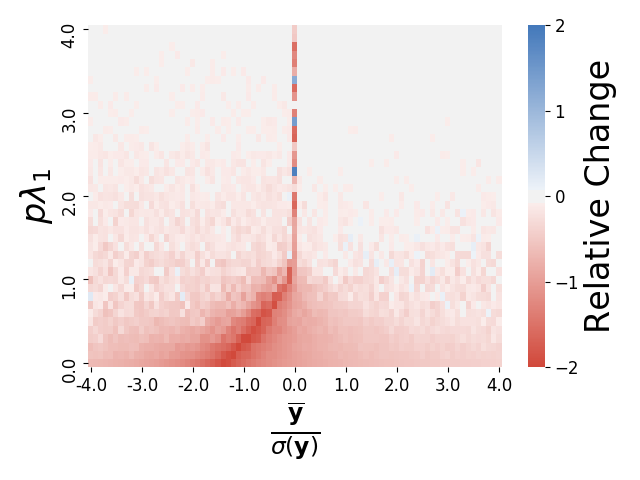} \label{fig:relative_change_degree_watts_strogatz}}
\subfloat[All models]{\includegraphics[width=.4\linewidth]{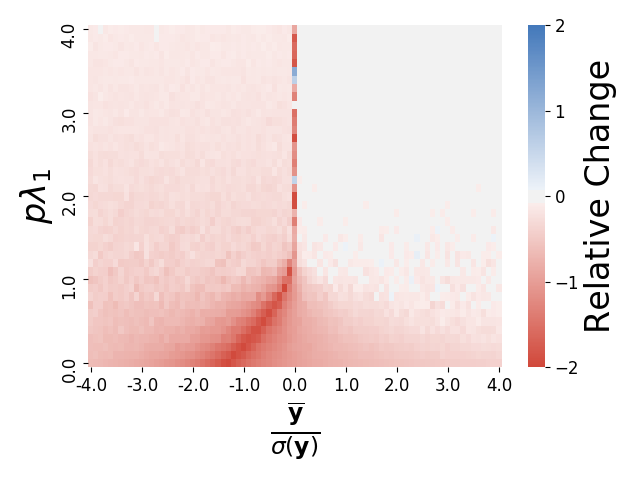}
\label{fig:relative_change_degree_all}}
\label{fig:perform_degree_synthetic}
\caption{Performance of degree centrality relative to other centralities on random networks. Each plot shows the relative change, computed as $\frac{a - b}{max(|a|, |b|)}$ where $a$ is degree centrality's average payoff and $b$ is the maximum average payoff of the other centralities, for varying values of $\frac{\overline{\mathbf{y}}}{\sigma(\mathbf{y})}$ and $p \lambda_1$.}
\end{figure*}

\begin{figure*}[ht!]
\centering
\subfloat[Barabasi-Albert model]{\includegraphics[width=.4\linewidth]{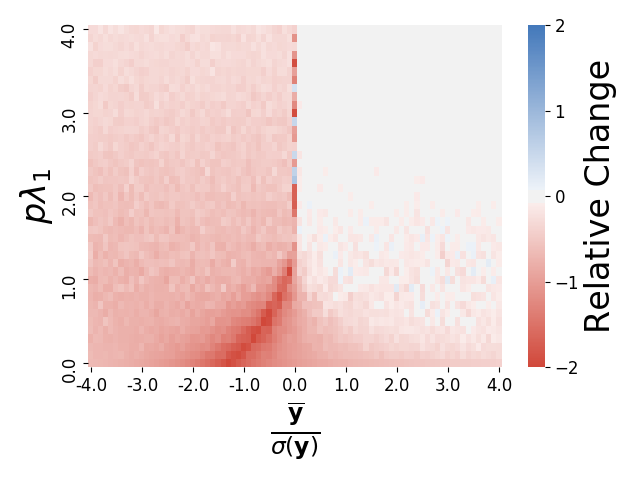} \label{fig:relative_change_diffusion_barabasi_albert}}
\subfloat[Erdos-Renyi model]{\includegraphics[width=.4\linewidth]{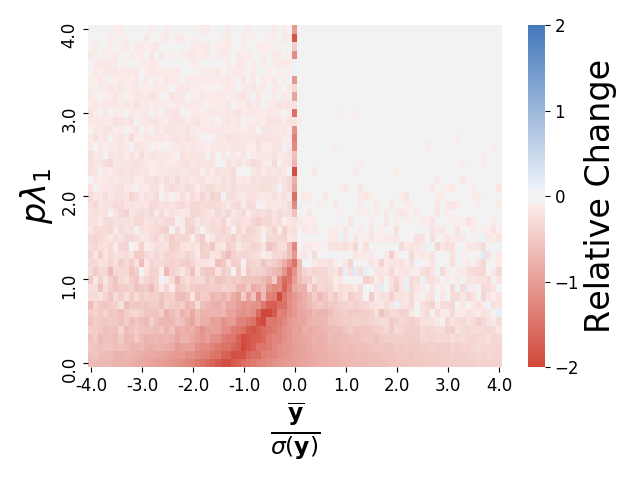} \label{fig:relative_change_diffusion_erdos_renyi}} \\
\subfloat[Watts-Strogatz model] {\includegraphics[width=.4\linewidth]{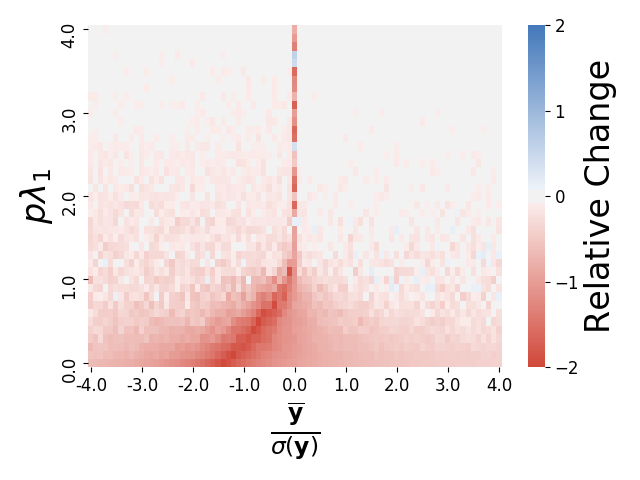} \label{fig:relative_change_diffusion_watts_strogatz}}
\subfloat[All models]{\includegraphics[width=.4\linewidth]{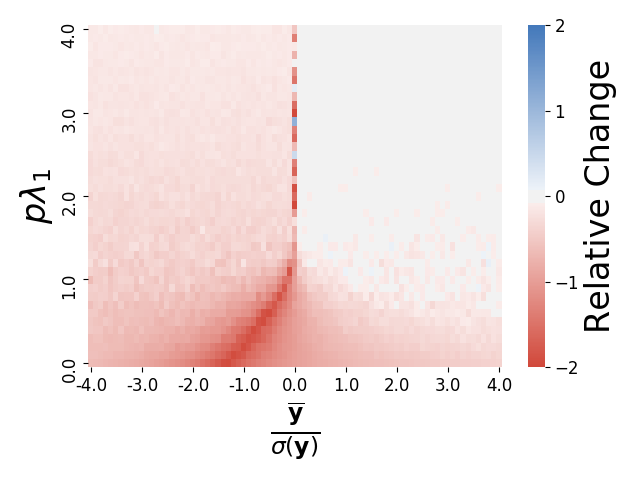}
\label{fig:relative_change_diffusion}}
\label{fig:perform_diffusion_synthetic}
\caption{Performance of diffusion centrality relative to other centralities on random networks. Each plot shows the relative change, computed as $\frac{a - b}{max(|a|, |b|)}$ where $a$ is diffusion centrality's average payoff and $b$ is the maximum average payoff of the other centralities, for varying values of $\frac{\overline{\mathbf{y}}}{\sigma(\mathbf{y})}$ and $p \lambda_1$.}
\end{figure*}

\begin{figure*}[ht!]
\centering
\subfloat[Barabasi-Albert model]{\includegraphics[width=.4\linewidth]{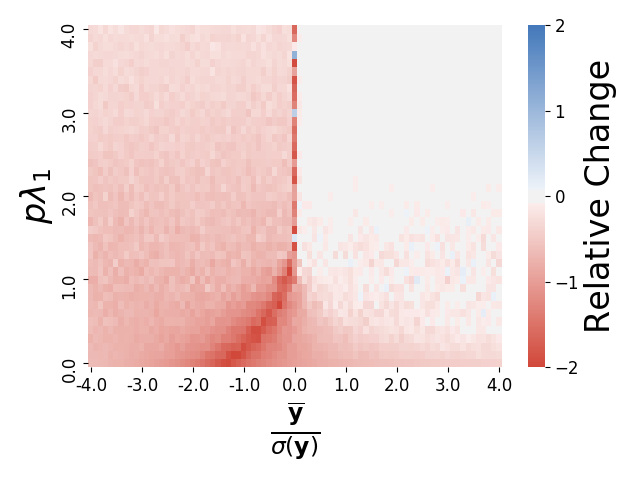} \label{fig:relative_change_eigenvector_barabasi_albert}}
\subfloat[Erdos-Renyi model]{\includegraphics[width=.4\linewidth]{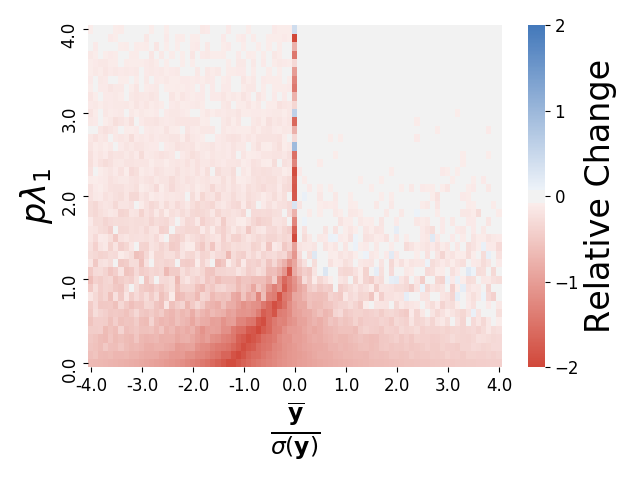} \label{fig:relative_change_eigenvector_erdos_renyi}} \\
\subfloat[Watts-Strogatz model] {\includegraphics[width=.4\linewidth]{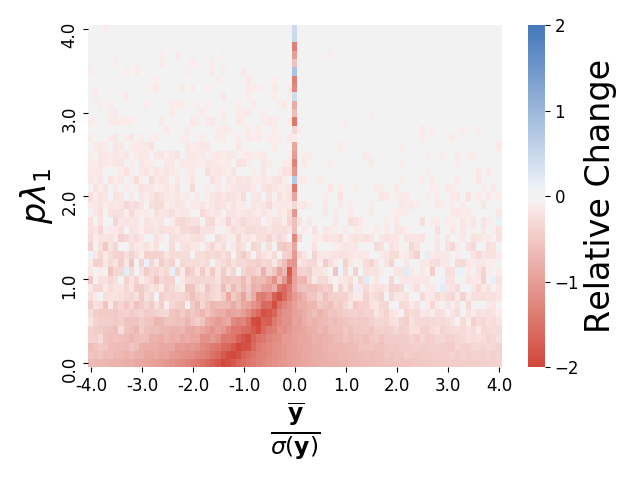} \label{fig:relative_change_eigenvector_watts_strogatz}}
\subfloat[All models]{\includegraphics[width=.4\linewidth]{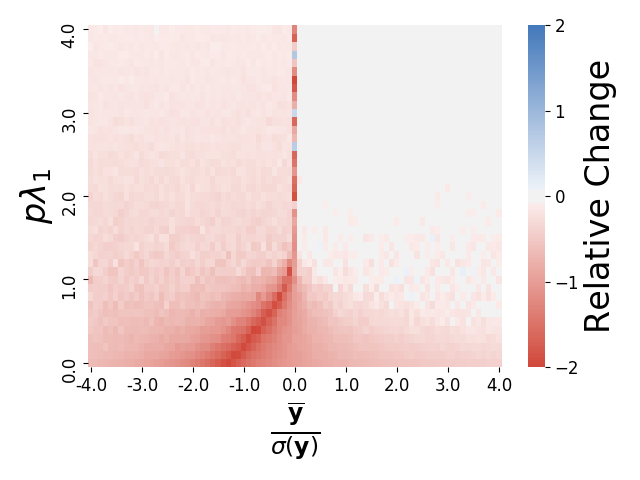}
\label{fig:relative_change_eigenvector}}
\label{fig:perform_eigenvector_synthetic}
\caption{Performance of eigenvector centrality relative to other centralities on random networks. Each plot shows the relative change, computed as $\frac{a - b}{max(|a|, |b|)}$ where $a$ is eigenvector centrality's average payoff and $b$ is the maximum average payoff of the other centralities, for varying values of $\frac{\overline{\mathbf{y}}
}{\sigma(\mathbf{y})}$ and $p \lambda_1$.}
\end{figure*}

\begin{figure*}[ht!]
\centering
\subfloat[Barabasi-Albert model]{\includegraphics[width=.4\linewidth]{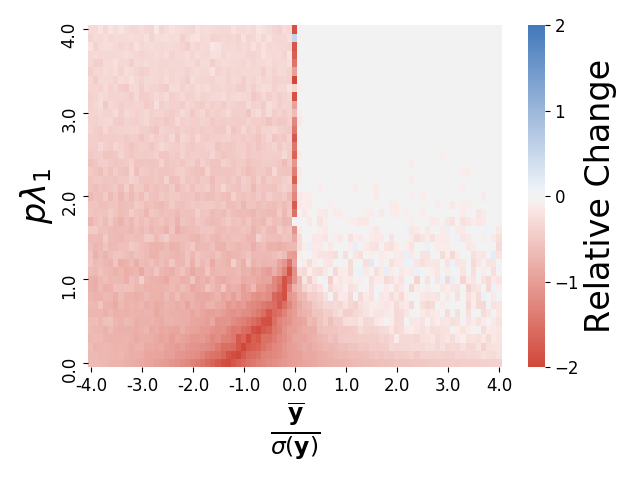} \label{fig:relative_change_katz_barabasi_albert}}
\subfloat[Erdos-Renyi model]{\includegraphics[width=.4\linewidth]{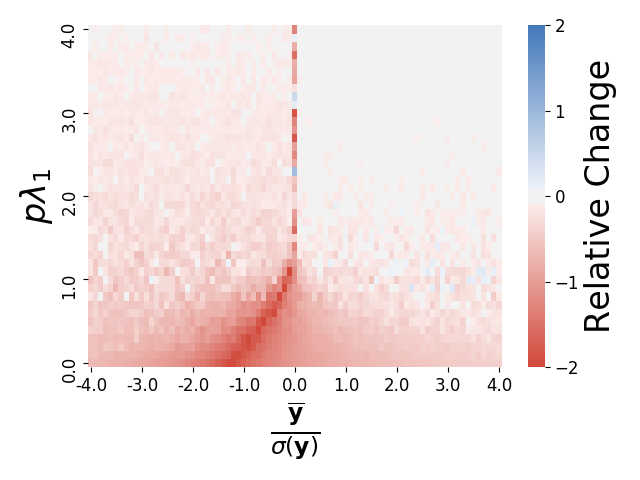} \label{fig:relative_change_katz_erdos_renyi}} \\
\subfloat[Watts-Strogatz model] {\includegraphics[width=.4\linewidth]{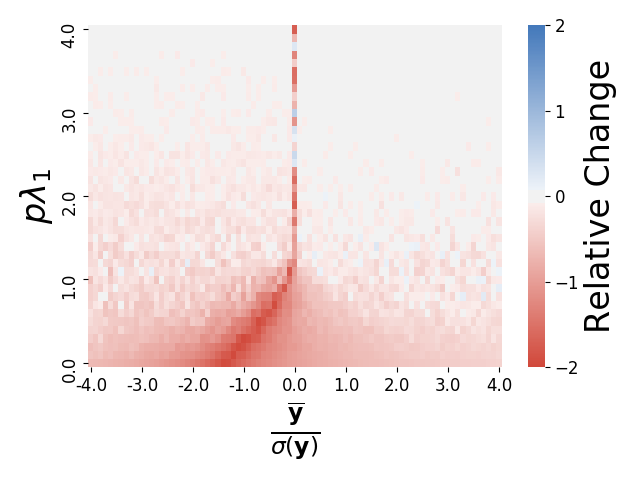} \label{fig:relative_change_katz_watts_strogatz}}
\subfloat[All models]{\includegraphics[width=.4\linewidth]{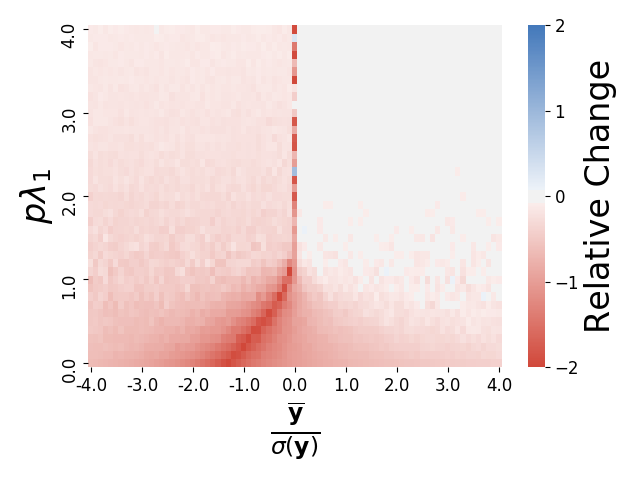}
\label{fig:relative_change_katz}}
\label{fig:perform_katz_synthetic}
\caption{Performance of Katz centrality relative to other centralities on random networks. Each plot shows the relative change, computed as $\frac{a - b}{max(|a|, |b|)}$ where $a$ is Katz centrality's average payoff and $b$ is the maximum average payoff of the other centralities, for varying values of $\frac{\overline{\mathbf{y}}}{\sigma(\mathbf{y})}$ and $p \lambda_1$.}
\end{figure*}

\clearpage

\subsection{Average approximated cascade payoff for contextual centrality and the variations of other centrality measures}

Here we present the average approximated cascade payoff for contextual centrality and the variations of other centrality measures. Note that the approximation does not hold for degree centrality when $p \lambda_1 > 1$ and $T$ is large. However, scaling degree centrality with primary contribution still improves the performance. Hence, we present it here.

\begin{figure*}[ht!]
\centering
\subfloat[degree]{\includegraphics[width=0.33\linewidth]{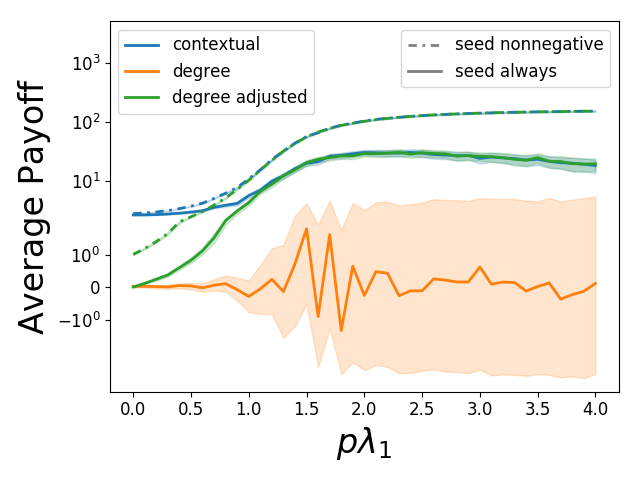}} 
\subfloat[katz]{\includegraphics[width=0.33\linewidth]{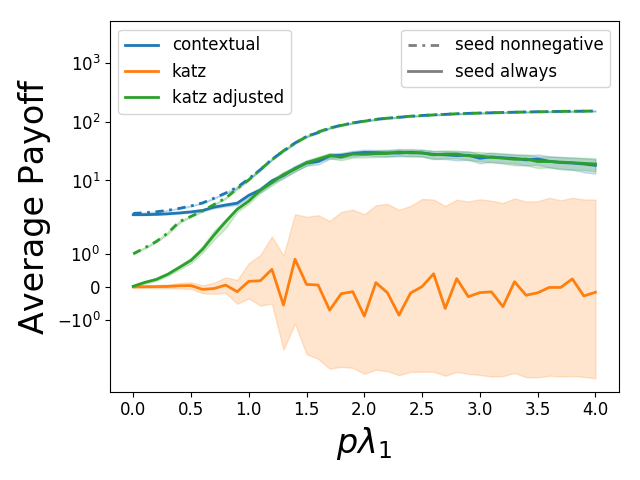}}
\subfloat[diffusion]{\includegraphics[width=0.33\linewidth]{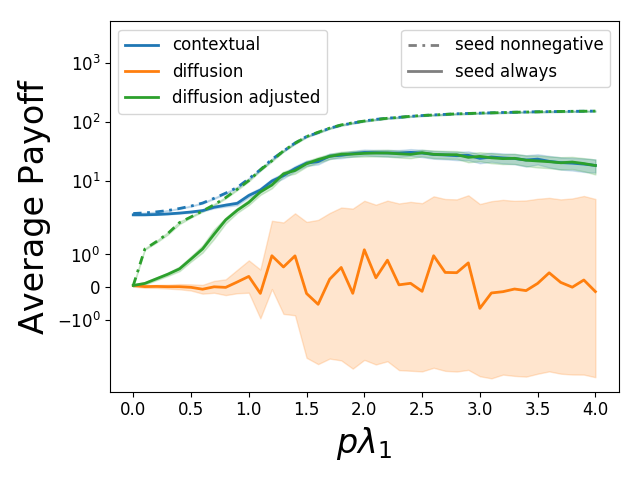}} 
\label{fig:average_payoff_all}
\caption{Average cascade payoff for contextual centrality and the variations of (a) degree, (b) diffusion, and (c) Katz centrality. The x-axis is $p \lambda_1$, and the y-axis is the average cascade payoff, with the shaded region as the 95\% confidence intervals. For `degree adjusted', `Katz adjusted' and `diffusion adjusted' centrality, we multiply the focal  centrality with the primary contribution $\mathbf{U}_1^T \mathbf{y}$. For `seed nonnegative', we adapt the original seeding strategy to seed only if the maximum of the centrality measure is nonnegative, otherwise it is named as `seed always'. }
\end{figure*}

\clearpage

\subsection{Comparision of seeding strategies when $\overline{\mathbf{y}} ( \mathbf{U}_1^T \mathbf{y}) < 0$}

Here we show the effect of using different seeding strategies on the average approximated cascade payoff. For this plot, we generated random networks as before with contributions sampled from a standard normal distribution, but redistributed the contributions to make the signs of $\overline{\mathbf{y}}$ and $\mathbf{U}_1^T \mathbf{y}$ differ if possible. More specifically, if the average contribution was negative, the individual with the largest eigenvector centrality score was given the most positive contribution, the individual with the second largest eigenvector centrality score was given the second most positive contribution, and so on. We used an analogous procedure if the average contribution was positive. Fig.~\ref{fig:average_payoff_seeding_strategies} shows that seeding according to the contextual centrality score tends to perform the best as long a $p \lambda_1$ is not too large, after which seeding according to the average contribution performs the best. For small values $p \lambda_1$, seeding always performs as well as, if not better than, seeding according to contextual centrality. As suggested by Eq.~(\ref{eq:prox_cc}), seeding according to the primary contribution yields similar results as seeding according to contextual centrality score as $p \lambda_1$ grows large.

\begin{figure}[ht!]
\centering
\includegraphics[width=.5\linewidth]{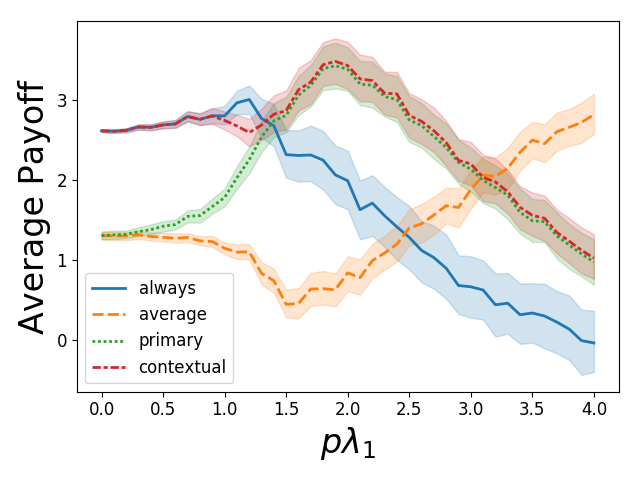}
\caption{Comparision of seeding strategies when $\overline{\mathbf{y}} ( \mathbf{U}_1^T \mathbf{y}) < 0$. Here we show the average approximated cascade payoff generated by seeding the top-ranked individual according to contextual centrality using different seeding strategies. The x-axis is $p \lambda_1$ and the y-axis is the average approximated cascade payoff with shaded 95\% confidence interval. The line marked ``always'' acts as our baseline in each we always seed the individual. For ``average'', we seed only if the average contribution is nonnegative. For ``primary'', we seed only if the primary contribution is nonnegative. For ``contextual'', we seed only if the contextual centrality score of the individual is nonnegative.} 
\label{fig:average_payoff_seeding_strategies}
\end{figure}

\end{document}